\documentclass[12pt]{article}

\usepackage{amsmath}
\usepackage{amssymb}
\usepackage{amsthm}
\usepackage{enumitem}
\usepackage[T1]{fontenc}
\usepackage{mathtools}
\usepackage{microtype}
\usepackage[dvipsnames]{xcolor}
\definecolor{darkgreen}{RGB}{0,100,0}
\usepackage[all,cmtip]{xy}
\usepackage{graphicx}
\newcommand{\R}{\mathbb{R}}
\renewcommand{\P}{\mathbb{P}}
\usepackage{verbatim}

\usepackage{biblatex}
\usepackage{hyperref}
\usepackage[capitalise]{cleveref}

\AtEveryBibitem{\ifboolexpr{
    test {\ifentrytype{article}}
    or
    test {\ifentrytype{book}}
  }{\clearfield{url}}{}}
\renewbibmacro{in:}{\setunit{\addcomma\space}}

\usepackage[format=plain,font={small,it}]{caption}
\usepackage{subcaption}
\usepackage{placeins}

\usepackage{algorithm}
\usepackage{algpseudocode}
\usepackage[normalem]{ulem}
\hypersetup{
  colorlinks=true,
  linkcolor=black,
  citecolor=black,
  filecolor=magenta,
  urlcolor=MidnightBlue
}

\usepackage[margin=1in]{geometry}

\theoremstyle{plain}
\newtheorem{theorem}{Theorem}[section]

\newtheorem{lemma}[theorem]{Lemma}

\theoremstyle{definition}
\newtheorem{definition}[theorem]{Definition}

\theoremstyle{remark}

\newtheorem{remark}[theorem]{Remark}

\newtheorem*{remark*}{Remark}
\newtheorem*{notation*}{Notation and Terminology}

\numberwithin{equation}{section}

\newcommand{\Romannum}[1]{\uppercase\expandafter{\romannumeral #1}}

\DeclarePairedDelimiterX{\pair}[2]{\langle}{\rangle}{#1,#2}

\renewcommand{\S}{Section }

\title{AI for AI: Optimizing Additional Infrastructure Build-out to
  Power Artificial Intelligence Data Centers}
\author{Alexander Crosier\thanks{Formerly Princeton University} \and Kyle Onghai\thanks{ORFE Department, Princeton University} \and Ronnie Sircar\thanks{\hangindent=1.75em
  \hangafter=1 ORFE Department, Bendheim Center for Finance and Andlinger Center for Energy \& the Environment,
Princeton University}}
\date{\today}

\begin{document}

\maketitle

\begin{abstract}
  The twenty-first century's transformative technology, artificial
  intelligence, is increasingly constrained by the twentieth century's
  transformative technology, the electricity grid. Rapid growth in
  electricity demand from data centers is leading to higher
  electricity prices, without a compensating supply-side response. We
  develop a framework linking data-center load growth, available
  generation capacity, and market-clearing prices to understand this
  phenomenon. We first analyze a deterministic model to show how
  differing estimates of demand and supply growth rates affect prices.
  We then model the expansion of new data centers and their associated
  electricity demand, together with build-outs of new electricity
  supply, as stochastic processes, resulting in probabilistic
  distributions of supply, demand, and prices rather than a single
  forecast. Finally, we formulate generation expansion as a stochastic
  control problem in which a revenue-maximizing investor dynamically
  chooses the intensity of supply-side investments. The analysis
  highlights a central challenge of the data-center build-out: even
  when rapid demand growth increases the need for new generation, the
  uncertainties related to load forecasts, development execution
  risks, and value cannibalization from overbuilding capacity may
  weaken incentives to invest at the pace required to keep electricity
  prices stable.
\end{abstract}

\section{Introduction}
\label{sec:introduction}

The artificial intelligence era has sparked a rapid increase in
electricity demand from data centers. From 2018 to 2025, electricity
demand from data centers alone increased from 76 TWh (1.9\% of U.S.
total demand) to more than 200 TWh (4.8\% of U.S. total demand)
\cite{shehabi2024datacenter}. Forecasts for future demand from AI data
centers vary widely, though grid operators and utilities across the
Midwest, Southeast, Southwest, and Texas expect these facilities to
increase peak demand by 20\% to 40\% or more over the next decade. In
2025, ERCOT, Texas's primary power grid, forecasted 33 GW of new
demand from data centers for artificial intelligence and
cryptocurrency mining by 2031, equivalent to a 36\% increase relative
to its all-time system peak of 91 GW \cite{ercot2025loadforecast}.
According to \cite{slok2026datacenter}, ``Texas alone accounts for
roughly 100 GW of planned data center capacity, more than the next two
states, Virginia and Utah, combined.''

These projections expose a fundamental tension between the rapid
development of artificial intelligence and the much slower expansion
of the infrastructure required to power data centers. In this sense,
one of the twenty-first century's most transformative technologies,
artificial intelligence, is increasingly constrained by one of the
twentieth century's most transformative technologies, the electricity
grid.

\subsection{Demand, Supply \& Price Uncertainties}
\label{sec:demand_supply_price_uncertainties}

On the supply side, interconnection queues, transmission constraints,
and supply chain challenges may limit the pace at which new generation
can be added to the grid. As aging thermal plants retire, replacement
resources may not connect quickly enough to meet the sudden surge in
electricity demand. Lead times for new gas turbines have reached five
to seven years, while renewable energy projects, the fastest way to
add new capacity, often face permitting hurdles, political headwinds,
and community opposition that can delay deployment for years
\cite{anderson2025gasturbines}. Calvin Butler, the chief executive of
the nation's largest utility, Exelon, warned that capacity-constrained
regions could be at risk of rolling blackouts as soon as 2027, and
rate increases for all customers would be needed to fund new
infrastructure \cite{muir2026blackouts}.

The challenge is not only how much new electricity supply should be
built, but also which technologies should be built and when. Natural
gas, solar, wind, large-scale nuclear, and small modular reactors
differ substantially in their capital and operating costs, reliability
characteristics, and permitting and supply-chain constraints. The
attractiveness of each option depends both on a technology's cost,
{\em and} how long it can take for it to be built and get connected to
the grid.

These issues have recently become politically prominent. PJM
Interconnection operates the grid across 13 states and the District of
Columbia, including regions where retail customers have recently
experienced substantial increases in electricity bills. Following a
20\% increase in residential power bills in parts of New Jersey, both
candidates in the state's 2025 gubernatorial race made consumers'
electricity costs a centerpiece of their campaigns
\cite{weil2025electricbills}. PJM's coverage includes the largest
concentration of data centers in the United States, the so-called Data
Center Alley in Northern Virginia, and the region's independent market
monitor attributed 40\% of the costs in the December 2025 capacity
auction to data centers expected to come online in 2027 and 2028
\cite{monitoringanalytics2026pjm}.

PJM's forecast demand for the 2028-9 delivery year increased by
roughly 2~GW, largely because of data-center development, while its
most recent capacity auction in July 2026 attracted only about 525~MW
of new resources for that delivery year \cite{utilitydive2026pjm}.
Capacity costs are only one component of retail bills, which combine
energy costs, capacity charges, transmission and distribution costs.
Nevertheless, higher wholesale energy and capacity prices ultimately
raise costs borne by households and businesses that are not directly
responsible for the new data-center load.

The economic mechanism is not simply that high prices call forth
instantaneous new supply. High electricity and capacity prices may be
necessary to attract generation investment, but they need not induce
sufficient or timely entry when build-outs are costly, slow,
uncertain, and irreversible. Moreover, the same investment that
alleviates scarcity erodes the revenues that motivated it. Because
electricity prices are determined by supply-demand market clearing,
each new completed supply generator lowers the price received by other
generators, thus cannibalizing the revenue earned by existing
capacity.

A capacity expansion decision weighs the value of one more completion
against this self-inflicted price decline, and when the latter
dominates, optimal investment can cease even while aggregate
electricity demand continues to grow. This tension between the
scarcity prices and the pace of new entry those prices induce,
compounded by the uncertainty and lumpy arrival of new build-outs, is
central to the stochastic model developed below.

\subsection{Contributions}
\label{sec:contributions}

Building upon the study of data-center impact on electricity grids
started in May 2025 coauthored by the first and third authors
\cite{sirik2026datacenters}, we develop a framework linking
data-center load growth, available electricity supply, and
market-clearing prices. Our starting point is a deterministic model
that tracks the average evolution of supply and demand (Section
\ref{sec:deterministic_growth}). It describes the path obtained when
electricity supply and data-center load grow at their forecasted
rates. A first stochastic model then preserves these underlying growth
rates while replacing smooth average growth with uncertain, discrete
arrivals (Section \ref{sec:stochastic_supply_demand}). The
deterministic analysis can be viewed as the mean-path counterpart of
this stochastic model: it describes where supply, demand, and prices
tend to move on average, while the stochastic formulation reveals the
distribution of outcomes around that benchmark.

We subsequently endogenize generation expansion from the perspective
of a revenue-maximizing capacity developer who dynamically chooses the
intensity of supply-side investment under uncertainty about both
data-center load growth and build-out completion (Section
\ref{sec:control_intensity}). This framework also allows investors to
compare alternative supply technologies, including natural gas, solar,
wind, and nuclear, whose attractiveness depends on their capital costs
and expected time to build-out. The resulting controlled-intensity
formulation captures both the uncertain arrival of new capacity and
the revenue-cannibalization effect through which successful investment
may lower the market-clearing price earned on the owner's generation
portfolio.
 \section{Electricity Demand, Supply, and Price Formation}
\label{sec:demand_supply_price}

We introduce a stylized dynamic model of an electricity market
historically organized around traditional residential, commercial, and
industrial consumers, but now being reshaped by the rapid arrival of
hyperscaler data centers. We therefore distinguish between traditional
(group 1) and hyperscaler (group 2) electricity demand, the latter of
which we treat as a relatively new, fast-growing and unpredictable
load driver.

\subsection{Market Clearing}
\label{sec:market_clearing}

We fix a reference price $P_0$ representative of the average wholesale
market price before anticipated rapid data center growth. The three
components capturing the state of a simplified electricity market at
time $t\geq0$ in our model are:
\begin{itemize}\itemsep -2pt
  \item \(S=(S_t)_{t\geq0}\) -- total accredited power generation capacity (supply);
  \item \(I=(I_t)_{t\geq0}\) -- traditional  (group 1) electricity demand at the reference price \(P_0\);
  \item \(X=(X_t)_{t\geq0}\) -- data-center (group 2) electricity demand at the reference price \(P_0\).
\end{itemize}
We simulate our model on the ERCOT-operated grid by taking parameters
inferred from published figures and estimates specific to Texas. More
specifically, we set the reference price \(P_0 = \$30\)/MWh to roughly
match ERCOT's average wholesale energy price \cite{ERCOT2024Overview},
initial reference aggregate traditional demand \(I_0 = 55\)~GW, and
reference aggregate data-center demand \(X_0 = 8\)~GW
\cite{coblermarket}.

\paragraph{Supply}
As time evolves, more supply capacity becomes available, and we will
assume $S_t$ is overall increasing (ignoring asset retirements,
degradation, and seasonal effects as relatively minor decrements). In
this stylized model, \(S_t\) represents total generation capacity that
is available to service demand. We do not model short-term events
affecting dispatchable energy like intermittency, outages,
transmission constraints when calculating this component as our goal
is to understand long-term shifts in grid operations.

\paragraph{Demand}
Increase in \(I_t\) over time reflects traditional (non-AI) economic
growth and broader electrification, while growth in \(X_t\) reflects
additional reference data-center demand. Price-responsive demand from
each group is obtained by multiplying their reference demand by their
respective demand elasticity (or price-response) functions $F_{1,2}$
at price level $P_t$:
\begin{equation}
  \label{eq:demand_general}
  D_1(I_t, P_t) = I_t F_1(P_t), \qquad  D_2(X_t, P_t) = X_t F_2(P_t).
\end{equation}
The demand elasticity functions \(F_{1,2}:\R_{\geq 0}\to\R_{\geq 0}\)
are positive (or zero), decreasing (or flat). The total demand
\[
  D(I_t,X_t,P_t)= D_1(I_t, P_t) + D_2(X_t, P_t)
\]
is therefore also a non-increasing, non-negative function of price for
each $I_t,X_t>0$. We impose the normalization
\begin{equation}
  \label{eq:normalize}
  F_1(P_0)=F_2(P_0)=1,
\end{equation}
which is consistent with the interpretations of \(I\) and \(X\) given
above. This gives that the demand at time $t=0$ is
$D(I_0,X_0,P_0) = I_0+X_0$. Thus \(I_t\) and \(X_t\) are reference
demand levels at the price \(P_0\), while \(D_1\) and \(D_2\) are the
corresponding price-responsive demands at the prevailing price
\(P_t\).

\paragraph{Price}
Given supply $S_t$, the price \(P_t\) is determined by the
market-clearing condition
\begin{equation}
  \label{eq:price_market_clearing}
  D(I_t,X_t,P_t)= S_t,
\end{equation}
with the proviso that if $S_t\geq D(I_t,X_t,0)$, we set $P_t=0$. We
think of $P_t$ as a long-term (for instance, monthly) average price at
time $t$ rather than a specific day-ahead or real-time price, both of
which fluctuates on a short (minute or hourly) microscopic time scale.
Thus we interpret any fluctuation in the price that the model produces
as a shift in the prevailing average price induced by a shift in
demand or supply.
\begin{remark}
  \label{rem:unique_price_condition}
  If \(F_1\) and \(F_2\) are \(C^1\) on an open interval
  \(O=(\underline p,\overline p)\subseteq \mathbb{R}_+\) and
  \(F_1',F_2'<0\) on \(O\), then there is at most one market-clearing
  price in \(O\). Indeed, define \(G(s,i,x,p)=s-D(i,x,p)\) so that the
  market-clearing condition \eqref{eq:price_market_clearing} becomes
  \(G(s,i,x,p)=0\). Observe that
  \(G\in C^1(\mathbb{R}_+^3\times O;\mathbb{R})\) and
  \[
    \frac{\partial G}{\partial p}(s,i,x,p) = -\frac{\partial
      D}{\partial p}(i,x,p) = -\left[iF_1'(p)+xF_2'(p)\right] >0
  \]
  for all \((s,i,x,p)\in\mathbb{R}_+^3\times O\). Because \(G\) is
  strictly increasing with respect to \ \(p\), there is at most one
  \(p^*\in O\) that clears the market. For existence, define
  \(D^+(i,x)=\lim_{p\uparrow \overline p}D(i,x,p)\),
  \(D^-(i,x)=\lim_{p\downarrow \underline p}D(i,x,p)\), and
  \(U=\{(s,i,x)\in\mathbb{R}_+^3:D^+(i,x)<s<D^-(i,x)\}\) to be the
  open domain of states where market-clearing is possible within
  \(O\). Since \(\lim_{p\downarrow \underline p}G(s,i,x,p)<0\) and
  \(\lim_{p\uparrow \overline p}G(s,i,x,p)>0\) for each
  \((s,i,x)\in U\), the intermediate value theorem furnishes a unique
  root \(p^*\in O\). In this way, we obtain a function \(p^*:U\to O\)
  such that
  \[
    (s,i,x)\in U \implies G(s,i,x,p^*(s,i,x))=0.
  \]
  Moreover, by the implicit function theorem, we have
  \(\partial_s p^*<0,\ \partial_x p^*>0,\ \partial_i p^*>0\). The
  signs of the partial derivatives confirm standard expectations that
  price decreases with supply and increases with demand.
\end{remark}

\subsection{Demand Response Functions}
\label{sec:demand_response_functions}

We now fix a particular choice of demand elasticity functions that we
will use throughout. These are constructed so that higher prices,
while causing both consumer groups to decrease their demand, impact
the hyperscalers to a lesser extent than the traditional users.

Let
\begin{equation}
  \label{eq:choke_price_elasticities}
  F_1(P) \coloneqq B_1 \left[ 1-\frac{P}{A_1} \right]^+, \qquad F_2(P) \coloneqq
  B_2\left(\left[ 1-\frac{P}{A_2} \right]^+\right)^2,
\end{equation}
where $A_{1,2}<\infty$ are the \textit{(effective) choke prices} of
traditional and hyperscale consumers respectively. We assume that the
reference price $P_0<A_{1,2}$ and set
\[
  B_1=(1-P_0/A_1)^{-1}>1, \qquad  B_2=(1-P_0/A_2)^{-2}>1,
  \]
so that the demand-response functions are normalized to one for both
groups at \(P_0\), consistent with \eqref{eq:normalize}.
Remark~\ref{rem:unique_price_condition} applies with
\(O = (0, \min\{A_1, A_2\})\).

In our simulation results, we retain \(P_0=\$30/\mathrm{MWh}\) as
motivated in Section~\ref{sec:market_clearing} and choose
$A_1=\$70/$MWh and $A_2=\$150/$MWh. The demand-response functions are
plotted in Figure~\ref{fig:demand_multiplier_curves}.
\begin{figure}[htb]
  \centering
  \includegraphics[width=0.7\textwidth]{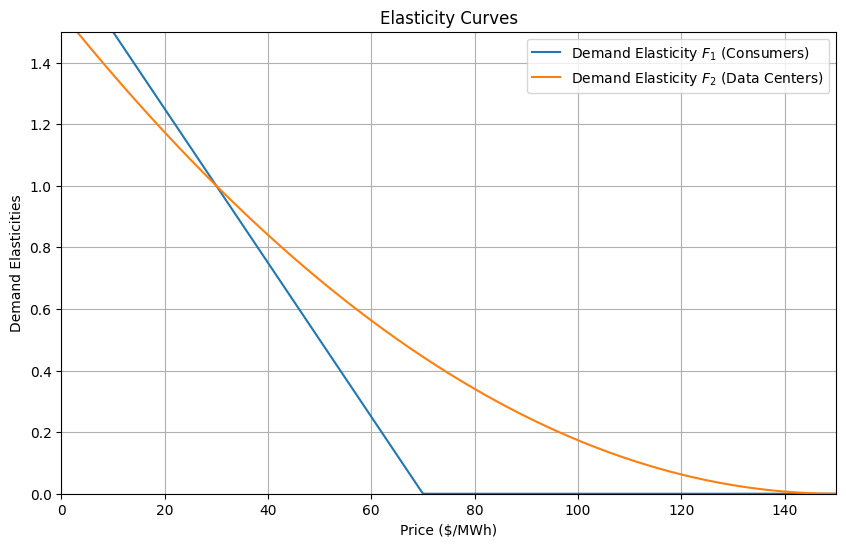}
  \caption{Demand elasticities for traditional consumers and
    hyperscalers normalized to $1$ at the reference price
    \(P_0=\$30\)/MWh, with $A_1=\$70/$MWh and
    $A_2=\$150/$MWh.}
  \label{fig:demand_multiplier_curves}
\end{figure}
With these parameters, traditional consumers will largely shut down
operations or move elsewhere if typical (not just peak) wholesale
prices approach $A_1 = \$70$/MWh while hyperscalers retain a higher
proportion of reference demand at these prices. However, we
demonstrate in Section~\ref{sec:deterministic_growth} that even when
the supply remains fixed at \(S_0 = I_0 + X_0\) with \(I_0 = 55\)~GW
and \(X_0 = 8\)~GW, the model's price does not reach $A_1$ until over
22 years, long after the time horizon \(T = 6\) years that we
consider. As \(S_t\) is increasing, we observe that neither group's
aggregate demand vanishes in both the deterministic
(Section~\ref{sec:deterministic_growth}) and stochastic
(Section~\ref{sec:stochastic_supply_demand}) settings.

\subsection{Deterministic Demand \& Supply Growth}
\label{sec:deterministic_growth}

We assume throughout that reference traditional demand grows at a rate
$\gamma\geq0$ relative to its size:
\begin{equation}
  \label{eq:Iteqn}
  I_t = I_0 e^{\gamma t},
\end{equation}
with running time $t$ measured in years. So, for example,
$\gamma=0.03$ would quantify an economy and electrification-driven
growth in reference traditional demand of $3\%$.

\subsubsection{No New Build-out}
\label{sec:no_new}

To give some intuition about the deterministic model's predictions, we
begin with the extreme case where supply stays fixed: $S_t=S_0$ for
all $t\geq0$. In other words, there is no additional infrastructure
built or planned for. This is in some ways our worst-case scenario, as
we exclude supply actually declining overall in our models. Detailed
calculations and figures are relegated to
Appendix~\ref{app:fixed_supply}.

\paragraph{Pre-hyperscaler Era}
Suppose first that there are no hyperscalers, so that \(X_t=0\) for
all \(t\geq 0\). Then the market clearing condition
\eqref{eq:price_market_clearing} reduces to $I_t F_1(P_t)=S_0$, and
substituting the demand-response functions established in Section
\ref{sec:demand_response_functions} leads to
$P_t =A_1-(A_1-P_0)e^{-\gamma t}$. Thus, in the absence of
hyperscalers, the price rises gradually from \(P_0<A_1\) but never
reaches \(A_1\). The rising market-clearing price adjusts the growing
reference demands to the fixed available supply. See
Figure~\ref{fig:hyperscaler_price_perturbation} (blue curve).

\paragraph{Linear data-center reference-demand growth}
The rise of data center demand changes this price trajectory. Suppose
next that \(X_t\) grows linearly over time: $ X_t = X_0 + c_X t$,
where \(c_X>0\) is the growth rate of reference data-center demand,
but that there is still no growth in supply. As a compromise between
ERCOT's adjusted large load breakdown and the TSP provided large load
breakdown, we choose \(c_X\coloneqq 6\)~GW/year for our simulations
\cite{ercot2025loadforecast}. The resulting price path is computed in
Appendix \ref{app:fixed_supply} and is also plotted in
Figure~\ref{fig:hyperscaler_price_perturbation} (red curve). Because
data-center demand is less price-sensitive, growth in reference
data-center demand raises the market-clearing price significantly
above the no-hyperscaler case.

Figure \ref{fig:fixed_supply_demand_decomposition} contrasts reference
and price-responsive demand for both groups, and illustrates the
crowding out of traditional demand by reference data-center demand
under fixed supply. We conclude that, even with no growth in supply,
traditional demand would not fall to zero until over 22 years, and
typical electricity prices would take that long to approach the group
1 choke price. However, without supply-side response, that is,
additional infrastructure, prices would double from $\$30$/MWh to
$\$60$/MWh over 10 years.

\subsubsection{Linear Supply Growth}
\label{sec:deterministic_supply_growth}

Suppose that available supply also grows linearly in time:
$S_t=S_0+c_St$, where \(c_S\geq 0\) is the amount of new supply
capacity added per year in GW/year.
Figure~\ref{fig:deterministic_supply_growth} shows the resulting price
paths for different choices of $c_S$. When supply grows slowly, demand
growth dominates and prices rise; as \(c_S\) increases, this upward
pressure is reduced, and sufficiently rapid supply growth causes
prices to fall, illustrating the cannibalization effect.
\begin{figure}[htbp]
  \centering
  \includegraphics[width=0.78\linewidth]{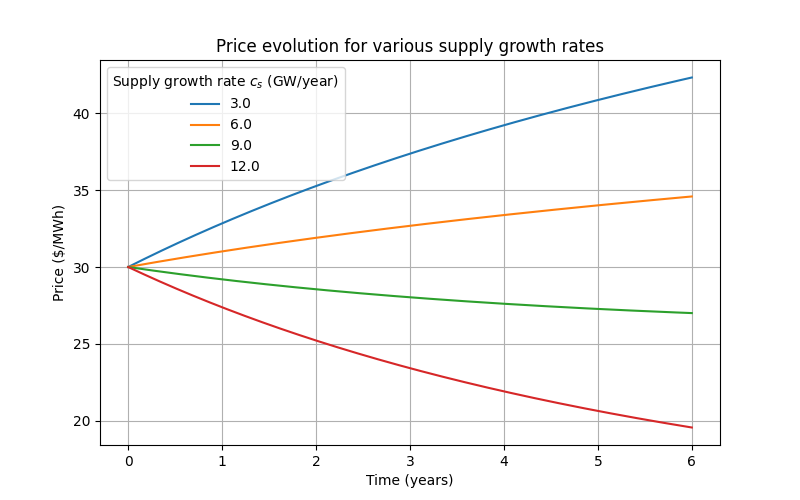}
  \caption{Average annual wholesale price with supply growth rates
    \(3, 6, 9, 12\) GW/year over 2025--2031.  General demand grows at
  3\% per year and reference data-center demand grows at 6 GW/year.}
  \label{fig:deterministic_supply_growth}
\end{figure}

While we could as a next step incorporate the costs of building new
supply and build a revenue-maximization model to quantify the optimal
build-out rate $c_S$, we do that instead in Section
\ref{sec:control_intensity}, in a stochastic framework which includes
uncertainties in supply and hyperscaler demand, which we introduce in
Section \ref{sec:stochastic_supply_demand}.

\subsection{Stochastic Supply and Demand Model}
\label{sec:stochastic_supply_demand}

The deterministic analysis in Section \ref{sec:deterministic_growth}
describes how electricity prices evolve when supply and demand grow
smoothly at specified average rates. In practice, however, neither new
generation capacity nor data-center load arrives smoothly. Power
plants come online at uncertain times, while new data centers create
large, irregular increments in electricity demand. We therefore
replace the deterministic growth paths with stochastic jump processes
that preserve the same underlying average growth rates while allowing
the timing of additions to be uncertain.

\subsubsection{Poisson and Compound Poisson Processes}
\label{sec:poisson_processes}

A Poisson process $N_t^\mu$, with intensity parameter $\mu>0$, is a
useful building block to model new data centers coming online at
uncertain times. It is a counting process, meaning it takes
successively values in $\{0,1,2,\cdots\}$, starting at zero:
$N^{\mu}_0=0$. Over a short time interval \([t,t+\Delta t]\) of length
$\Delta t$, the probability that $N^\mu$ jumps by one is approximately
$\mu\Delta t$:
\[
  \P\{N^\mu_{t+\Delta t} - N^\mu_t=1\} = \mu\Delta t+o(\Delta t),
\]
while the probability of no jump is
\[
  \P\{N^\mu_{t+\Delta t} - N^\mu_t=0\} = 1-\mu\Delta t+o(\Delta t).
\]
Consequently, the probability of two or more jumps over the short time
period \textcolor{darkgreen}{is} negligibly small.

A compound Poisson process extends this framework by allowing each new
data center to have a random reference-demand increment. If
$Y_1,Y_2,\ldots$ are independent and identically distributed jump
sizes, then
\begin{equation}
  \label{eq:Xdef}
  X_t=X_0+\sum_{k=1}^{N_t^\mu}Y_k
\end{equation}
is called a compound Poisson process. The counting process $N_t^\mu$
determines when new data centers come online, while $Y_k$ determines
the increase in reference data-center demand contributed by the
\(k\)th arriving data center.

\subsubsection{Demand Model \& Parameters}
\label{sec:demand_params}

On the demand side, we keep the deterministic growth formula
\eqref{eq:Iteqn} for \(I_t\), and we model reference data-center demand
\(X_t\) as a compound Poisson process \eqref{eq:Xdef}. According to
recent reporting by the Texas Tribune \cite{TexasTribune_ERCOT_2032},
ERCOT projects that peak demand on the Texas grid could reach 175~GW
by 2032, nearly double the state's current record peak. This explosive
growth is largely attributed to the rapid expansion of data centers
handling AI-related workloads.
Figure~\ref{fig:operating_and_planned_histogram} shows a
representative distribution of data-center campuses across capacity
bins.
\begin{figure}[htb]
  \centering
  \includegraphics[width=0.7\textwidth]{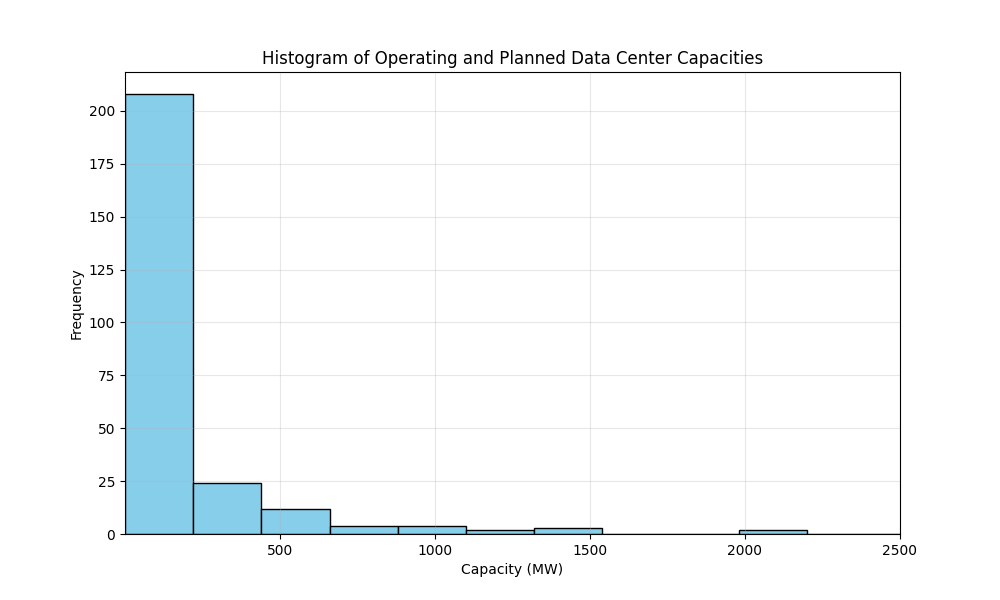}
  \caption{\looseness=0 Distribution of Texas data-center campus capacities
    (operating and planned). Most campuses fall below 0.5~GW, while a
    small number of very large projects create a pronounced right
    tail. Data centers with capacity above 2.5~GW (e.g. Fermi Project
    Matador) are not displayed. Source: \cite{ford_data_2026}.}
  \label{fig:operating_and_planned_histogram}
\end{figure}
The distribution is strongly right-skewed, reflecting that while many
facilities are smaller, the massive aggregate growth is driven by the
addition of extraordinarily large facilities. This unprecedented scale
and empirical pattern motivate a model in which very large data-center
additions are rare but possible.

Additions to reference data-center demand are thus random variables:
each new data-center arrival contributes a reference-demand increment
$Y_k$ drawn from a lognormal distribution calibrated by
maximum-likelihood on Texas data center capacities without the Project
Matador outlier, giving mean \(0.225\) GW and standard deviation
\(1.54\) GW. The choice of lognormal is motivated by the observed
heavy right tail. Performing a Monte Carlo goodness-of-fit test as
specified in \cite{scipy_goodness_of_fit, 2020SciPy-NMeth}
demonstrates that the MLE-fitted lognormal distribution describes the
data center capacities (excluding the Project Matador outlier) well.

The data-center build-out intensity is set to
$\mu=6/0.225\approx 26$/year, so that $\mu\mathbb{E}[Y]=6$~GW/year
agrees with the choice of \(c_X\) in Section~\ref{sec:no_new}.

\subsubsection{Supply Model \& Parameters}
\label{sec:supply_params}

For additional infrastructure technologies $j=1,\ldots,d$ (see Table
\ref{tbl:historical_supply_parameters}), let $N_t^{\lambda_j}$ be
independent Poisson processes with intensities $\lambda_j>0$, where
each arrival represents the completion of an additional generation
build-out of technology $j$. Let $\delta_j$ denote the amount of
generation capacity added by one such installation, assumed constant
for each technology, for simplicity. The total available generation
capacity at time $t$ is then
\begin{equation}
  \label{eq:Steqn}
  S_t=S_0+\sum_{j=1}^d \delta_j\,N_t^{\lambda_j}.
\end{equation}
The technology-specific counting processes describe additions of new
infrastructure, while $S_t$ is the aggregate generation capacity
across all technologies entering the market-clearing condition
\eqref{eq:price_market_clearing}.

We distinguish additional infrastructure by generation technology:
natural gas, coal, solar, wind, large-scale (LS) nuclear, and small
modular reactor (SMR) nuclear. They differ both in terms of the
average build-out time (for instance long for LS nuclear, considerably
shorter for a new solar installation); and the increase in potential
generation capacity that they bring (large for LS nuclear, much
smaller for solar or wind farms). Using the additional infrastructure
capacity data in \cite{eia860m_dec2025}, we calibrate the Poisson
arrival intensities $\lambda_j$ and representative capacity increment
$\delta_j$ for each technology, as reported in
Table~\ref{tbl:historical_supply_parameters}.
\begin{table}[htb]
  \centering
  \begin{tabular}{|c|c|c|c|}
    \hline
    Technology $j$ & $\lambda_j$ (per year) & $\delta_j$ (MW) & \parbox{0.35\textwidth}{\centering Expected additional capacity $\lambda_j\delta_j$ (MW/year)} \\[2ex]
    \hline
    Natural gas & 10 & 250 & 2500\\
    Coal & 0.25 & 500 & 125\\
    Solar & 40 & 50 & 2000\\
    Wind & 10 & 100 & 1000\\
    LS nuclear & 0.1 & 1000 & 100\\
    SMR nuclear & 0.25 & 250 & 62.5\\
    \hline
  \end{tabular}
  \caption{Baseline parameters for stochastic capacity additions by
    generation technology. Parameters for established technologies are
    calibrated using 2015--2025 data from \cite{EIA_capacity}; the SMR
    parameters are assumed as described in the text.}
  \label{tbl:historical_supply_parameters}
\end{table}

As $\lambda_j$ is the expected number of capacity additions of
technology $j$ per year, while $\delta_j$ is the capacity added by
each arrival, $\lambda_j \delta_j$ is the expected annual contribution
of technology $j$ to total generation capacity. Because commercial
SMRs have not yet generated a historical record of capacity additions,
their parameters cannot be calibrated in the same way. We assume an
SMR arrival intensity corresponding to one new build every four years
on average.

\subsubsection{Simulations}
\label{sec:current_supply}

We compute $N=1000$ Monte Carlo simulations of the supply process
$S_t$ in \eqref{eq:Steqn} and the hyperscaler demand driver $X_t$ in
\eqref{eq:Xdef} over $T=6$ years, representing years 2025-31, summarizing
the results in Figure~\ref{fig:current_mc_paths_choke}. Three
representative paths of $S$ are shown in
panel~(\subref{fig:current_mc_S}), the first of which is divided more
granularly into generator technologies in
panel~(\subref{fig:current_mc_S_mix}). The three corresponding paths
of \(X\) are plotted in panel~(\subref{fig:current_mc_X}).

Along each path, we calculate the prices \(P_t\) from the
market-clearing condition \eqref{eq:price_market_clearing}.
Panel~(\subref{fig:current_mc_P}) exhibits the three corresponding
price paths, while panel~(\subref{fig:current_mc_terminal_price})
shows the full Monte Carlo distribution of terminal prices in 2031.

Quantitatively, the mean terminal price is \$33.43/MWh and the
standard deviation is \$5.87/MWh. In both figures, we emphasize the
possibility of prices falling below the baseline level \(P_0\). The
price-responsive demands in the current-supply scenario are shown in
panel~(\subref{fig:current_mc_true_demands}). Because prices rise
above the reference level when new data centers interconnect,
price-responsive demand shocks are more mild than the reference demand
trajectories suggest. Nevertheless, after accounting for electricity
prices, data centers extract a significant proportion of total
generation capacity. Across the full Monte Carlo sample, the mean
data-center share of price-responsive demand at the terminal date is
27.6\%.

\begin{figure}[htbp]
  \centering
  \begin{subfigure}{0.48\linewidth}
    \centering
    \includegraphics[width=\linewidth]{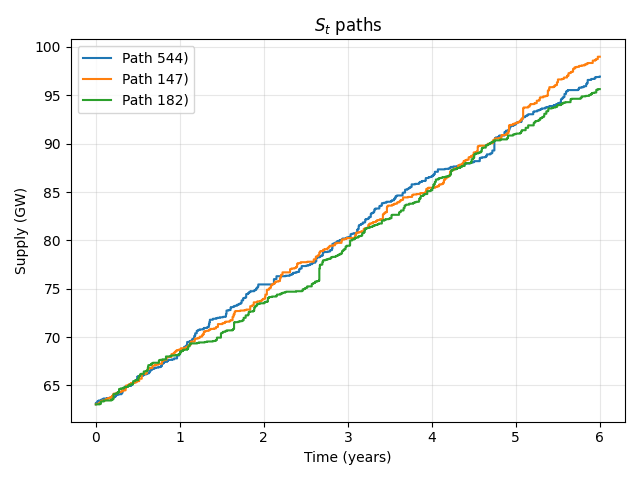}
    \caption{}
    \label{fig:current_mc_S}
  \end{subfigure}
  \begin{subfigure}{0.48\linewidth}
    \centering
    \includegraphics[width=\linewidth]{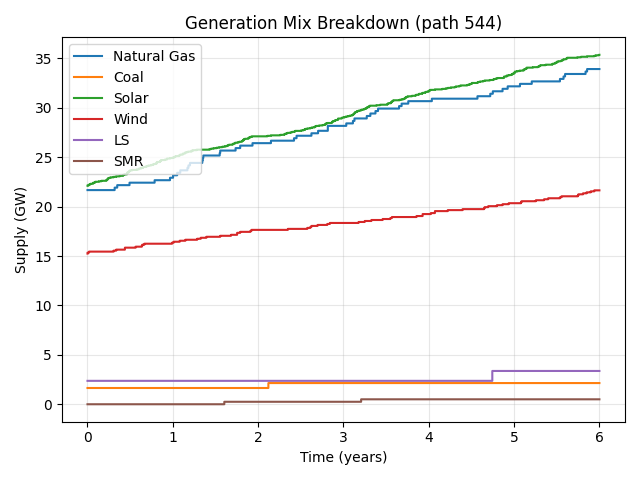}
    \caption{}
    \label{fig:current_mc_S_mix}
  \end{subfigure}
  \begin{subfigure}{0.48\linewidth}
    \centering
    \includegraphics[width=\linewidth]{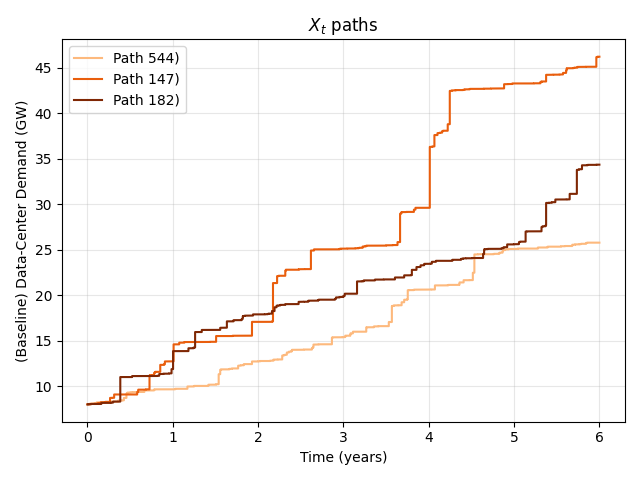}
    \caption{}
    \label{fig:current_mc_X}
  \end{subfigure}
  \begin{subfigure}{0.48\textwidth}
    \centering
    \includegraphics[width=\linewidth]{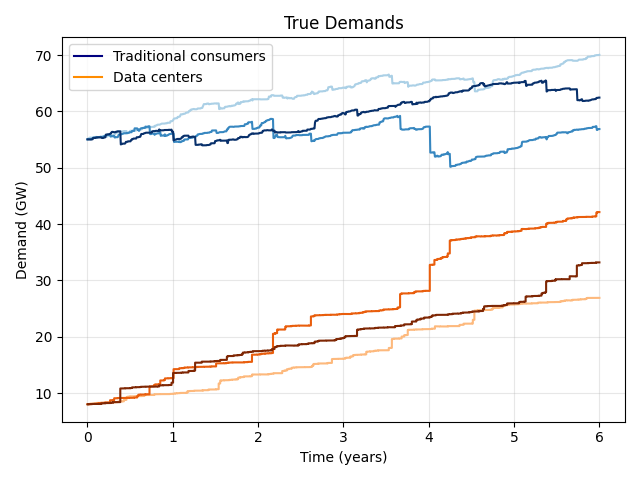}
    \caption{}
    \label{fig:current_mc_true_demands}
  \end{subfigure}
  \begin{subfigure}{0.48\linewidth}
    \centering
    \includegraphics[width=\linewidth]{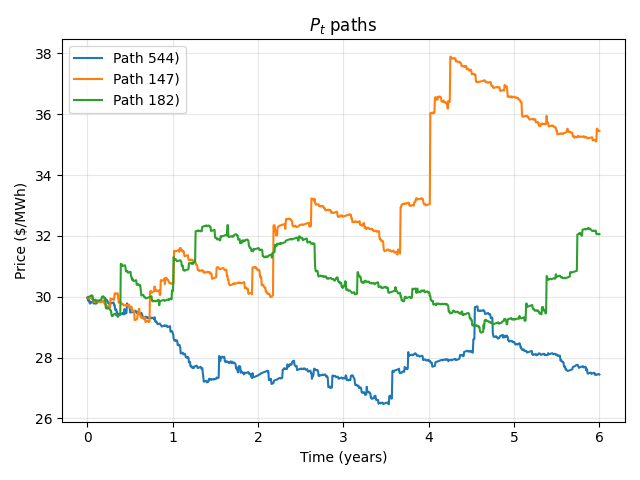}
    \caption{}
    \label{fig:current_mc_P}
  \end{subfigure}
  \begin{subfigure}{0.48\linewidth}
    \centering
    \includegraphics[width=\linewidth]{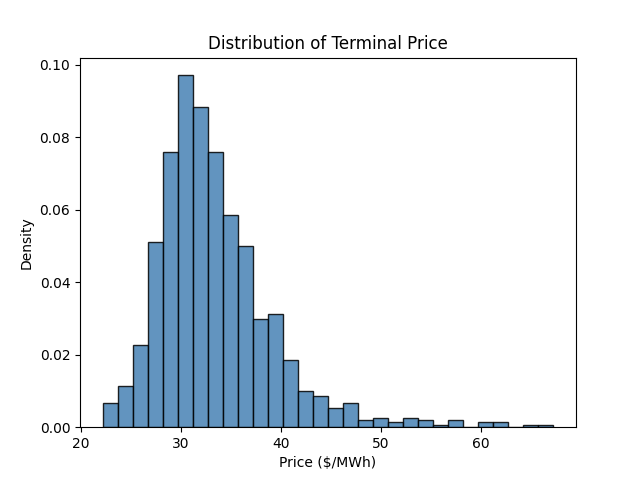}
    \caption{}
    \label{fig:current_mc_terminal_price}
  \end{subfigure}
  \caption{Summary of $1000$ Monte Carlo simulations in the current
    trend scenario (Section~\ref{sec:current_supply}). Top:
    (\subref{fig:current_mc_S}) three simulated paths of total
    generation capacity, and (\subref{fig:current_mc_S_mix}) breakdown
    of one of these capacity paths into generative technologies.
    Middle: (\subref{fig:current_mc_X}) three simulated paths of
    reference data-center demand, and
    (\subref{fig:current_mc_true_demands}) price-responsive demand of
    traditional consumers and data centers. Bottom:
    (\subref{fig:current_mc_P}) three representative price paths, and
    (\subref{fig:current_mc_terminal_price}) the terminal-price
    distribution.}
  \label{fig:current_mc_paths_choke}
\end{figure}

The deterministic growth model in Section
\ref{sec:deterministic_growth} identifies the average capacity growth
needed to stabilize prices, while the stochastic growth model of
Section \ref{sec:stochastic_supply_demand} shows how supply additions
and data-center load arrivals create dispersion around those averages.
Rapid supply expansion can suppress prices and create revenue
cannibalization as seen in panel~(\subref{fig:current_mc_P}), which
may weaken incentives to build capacity even when aggregate load is
rising. These effects motivate the controlled-intensity formulation in
\S\ref{sec:control_intensity}, where a generation investor or capacity
owner chooses supply expansion dynamically rather than it being
exogenously specified.
 \section{Optimal Supply-Side Investment}
\label{sec:control_intensity}

We now develop a stochastic control framework in which a generation
owner chooses the intensity of supply additions under uncertain
data-center load growth. Unlike the exogenous supply processes studied
in Section~\ref{sec:demand_supply_price}, capacity expansion is now
endogenous while retaining uncertain infrastructure completion times.
Such delays may be attributed to, for instance, supply-chain
disruptions or permitting delays. The owner controls the arrival
intensity of new generation capacity rather than its exact
installation date in order to maximize discounted her revenue net of
capacity-expansion costs up to a finite time horizon \(T = 6\),
matching the terminal time in Section~\ref{sec:current_supply}. In
this way, our representative generation owner is an investor whose
portfolio is the available supply \(S_t\). Investment thus affects
both the quantity of generation earning revenue and the
market-clearing price received on the existing portfolio, thereby
capturing the revenue-cannibalization effect directly. Our formulation
focuses on this aggregate investment incentive, leaving the strategic
interaction among individual generators to future work.

There is long-dated and large literature on stochastic control models
for irreversible capacity investment under uncertainty. Classical
approaches typically represent investment through singular controls,
in which capacity is added immediately when the investment decision is
made; see \cite{manne1961,dixit1994}. We instead represent build-out
completion through a controlled counting process, allowing the waiting
time until the next capacity addition to be random and state
dependent. The model also differs from approaches with exogenous
output prices because the electricity price is determined endogenously
by market clearing. Consequently, additional capacity affects both the
quantity available for sale and the price received on the owner's
generation portfolio.

Related controlled-intensity models have appeared in energy and
resource economics, where dynamic oligopoly models use intensity
controls to describe exploration and capacity expansion under
competition, for instance
\cite{LudkovskiSircar2012Exhaustibility,ChanSircar2017Fracking,
  HubertLolasSircar2025CapacityExpansion}, or cryptocurrency mining
\cite{LiReppenSircar2024CryptoMining}, or ticket pricing
\cite{AydinParmaksizSircar2025FareGame}. Our model adapts this
controlled-intensity framework to electricity capacity expansion by
combining uncertain project completion with uncertain data-center load
and an endogenous market-clearing price.

\subsection{Single-Technology Investment}
\label{sec:single_technology}

We first suppose supply is increased only by a single available
technology ($d=1)$, for instance wind. We incorporate multiple asset
types in Section \ref{sec:multiple_technologies}. The state variables
are the available generation capacity $S_t$ and reference data-center
demand \(X_t\), as introduced in
Section~\ref{sec:demand_supply_price}. The reference traditional
demand continues to evolve deterministically according to
$I_t = I_0 e^{\gamma t}$.

For tractability in the stochastic control problem, we simplify the
compound-Poisson specification for data-center load in
Section~\ref{sec:stochastic_supply_demand} by replacing the random
jump sizes $Y_k$ with a representative constant increment
$\kappa = \mathbb{E}[Y]$, while retaining stochastic arrival times.
Thus, reference data-center demand \(X\) evolves as
\begin{equation}
  \label{eq:data_center_process}
  dX_t=\kappa\,dN_t^\mu,
\end{equation}
where $N^\mu$ is a Poisson process with intensity $\mu>0$. On the
supply side, rather than taking the infrastructure build-out process
as exogenous as in Section~\ref{sec:stochastic_supply_demand}, the
investor now controls its arrival intensity:
\begin{equation}
  \label{eq:controlled_supply_process}
  dS_t=\delta\,dN_t^\lambda
\end{equation}
where $N^\lambda$ is a point process whose nonnegative intensity
$\lambda_t\geq0$ is chosen by the investor at each time
\(t\in [0, T]\). The precise admissibility and integrability
conditions are stated in Appendix~\ref{app:hjb}.

Roughly speaking, the effect of choosing intensity \(\lambda_t\geq0\)
at time \(t\) is to make the probability that $N^\lambda$ jumps by one
over a small time period \([t,t+\Delta t]\) of length $\Delta t$ is
approximately $\lambda\Delta t$:
\[
  \P\{N^\lambda_{t+\Delta t} - N^\lambda_t=1\} = \lambda_t\Delta t+o(\Delta t),
\]
while the probability of no increment is
\[
  \P\{N^\lambda_{t+\Delta t} - N^\lambda_t=0\} = 1-\lambda_t\Delta t+o(\Delta t).
\]
This makes the probability of two or more jumps over the short time
period negligibly small. Consequently,
\eqref{eq:controlled_supply_process} says that when $N^\lambda$
increases by one, supply capacity $S$ increases by $\delta$:
\begin{align*}
    \P\{S_{t+\Delta t} - S_t=\delta\} &= \lambda_t\Delta t+o(\Delta t),\\
    \P\{S_{t+\Delta t} - S_t=0\} &= 1-\lambda_t\Delta t+o(\Delta t).
\end{align*}
As such, $\lambda_t$ controls the probability of an immediate jump in
supply at time $t$. It can be viewed as a measure of investment and
effort to build a new generation asset: the larger the investment
intensity over \([t,t_1)\): $(\lambda_s)_{t\leq s<t_1}$, the greater
the likelihood of one (or more) new build-outs between times $t$ and
some $t_1>t$.

However, a higher supply intensity requires greater development effort
and incurs higher costs. We represent this by an increasing convex
cost function $C:\mathbb{R}_{\geq 0}\to\mathbb{R}_{\geq 0}$, where
\(C(\lambda_t)\Delta t\) is the instantaneous cost of maintaining
investment intensity \(\lambda_t\) over the small time interval
$[t,t+\Delta t)$. Having $C$ increasing captures greater effort
incurring greater cost, while convexity captures the concept of
diminishing returns, i.e. the marginal cost of investment is
increasing. We further assume $C(0)=0$ so that the supplier can stop
incurring costs by choosing intensity equal to zero and the technical
assumption \(\lambda\mapsto C(\lambda)\) is differentiable with
\(\lim_{\lambda\to \infty} C'(\lambda) = \infty\), so that the optimal
intensity is always uniquely defined. These are common assumptions in
the project investment and research \& development literature; see,
for instance, \cite{dixit1994}.

The running payoff is producer revenue net of the cost of maintaining
the chosen installation intensity: $S_tP(t,S_t,X_t)-C(\lambda_t)$. We
note that, because we record supply and demand quantities in GW, and
price in $\$/$MWh, the unit of revenue $S\times P$ is $\$1000/$hour,
and cost $C$ is assumed also to be in $\$1000/$hour. A realized supply
jump changes the state from \(S_t\) to \(S_t+\delta\), and hence
affects subsequent producer revenue through both available capacity
and the market-clearing price. Since $I_t$ is deterministic, the
controlled state is $(S_t,X_t)$, with time $t$ entering explicitly. We
stress the dependence of the market-clearing price \(P_t\) determined
by \eqref{eq:price_market_clearing} on the controlled state by writing
$P(t,S_t,X_t)$.

The generation investor maximizes over intensity
$(\lambda_t)_{t\in[0,T)}$ their expected discounted revenue minus
cost, up to a finite time horizon \(T\):
\begin{equation}
  \label{eq:value1}
  \mathbb{E}\left\{ \int_0^T e^{-ru}\,8760\times[S_u P(u, S_u, X_u) - C(\lambda_u)]\, \mathrm{d}u\right\},
\end{equation}
where the $8760$ hours/year adjust the time units of revenue and cost.
Here, $r>0$ is an annualized rate at which future profits are
discounted.

To use dynamic programming to solve this stochastic control problem,
we define the the value function $v:[0,T]\times\R_+\times\R_+\to\R_+$
by
\begin{equation}
  \label{eq:single_value_function}
  v(t, s, x) = \sup_{\lambda\in \Lambda} \mathbb{E}\left\{ \int_t^T e^{-r(u-t)}\left[S_u P(u, S_u, X_u) - C(\lambda_u)\right]\, \mathrm{d}u \mid S_t=s, X_t=x \right\},
\end{equation}
where \(\Lambda\) denotes the admissible class defined in
Appendix~\ref{app:hjb}, and we have divided by $8760$ in our
definition of $v$, which cancels the conversion factor in
\eqref{eq:value1}. We note that $v\geq0$ because prices $P$ and supply
$S$ are non-negative and doing nothing ($\lambda_t\equiv0$) is an
admissible costless strategy. The function $v$ encodes the optimal
value of potential future profits for the supply investor when the
program starts at time $t\in[0,T]$ with current supply level $s>0$ and
data-center demand at the reference price equal to $x>0$.

The dynamic programming principle gives the single-technology
Hamilton-Jacobi-Bellman (HJB) differential equation for $v$:
\begin{equation}
  \label{eq:hjb_single}
  \partial_t v(t,s,x) +\mu\Delta_xv(t,s,x) +\sup_{\lambda\geq0} \left\{\lambda\Delta_sv(t,s,x)-C(\lambda)\right\} +sP(t,s,x)-rv(t,s,x)=0,
\end{equation}
with $v(T,s,x)=0$.  The differences
\[
  \Delta_s v(t,s,x) \coloneqq v(t,s+\delta,x)-v(t,s,x), \qquad
  \Delta_x v(t,s,x) \coloneqq v(t,s,x+\kappa)-v(t,s,x),
\]
denote the changes in the value function resulting from one additional
supply-capacity completion and one additional data center arrival,
respectively. The optimal investment intensity is given by
\[
  \lambda^*(t,s,x)= \begin{cases}
    0, & \Delta_s v(t,s,x) \leq C'(0), \\[4pt]
    (C')^{-1}\!\left(\Delta_s v(t,s,x)\right), & \Delta_s v(t,s,x) > C'(0).
  \end{cases}
\]

To interpret \eqref{eq:hjb_single}, we state the role of each term.
The derivative \(\partial_t v\) records the passage of time;
\(\mu \Delta_x v\) captures the expected change in value arising from
a new increment in reference data-center demand;
\(\lambda \Delta_s v\) analogously measures the same for a new
generator; \(C(\lambda)\) is the cost explained above; current revenue
is \(sP\); and \(rv\) is a consequence of discounting future returns.
The derivation of equation \eqref{eq:hjb_single} and verification
theorem that its solution recovers the value of the stochastic control
problem \eqref{eq:single_value_function} are given in
Appendix~\ref{app:hjb}.

We work with a power-cost specification
\begin{equation}
  \label{eq:powercosts}
  C(\lambda) =\frac{1}{\beta}\lambda^\beta+\rho\lambda,
  \qquad \beta>1,\; \rho\geq0.
\end{equation}
If $\rho>0$ in \eqref{eq:powercosts}, then $C'(0)=\rho$ represents a
positive marginal cost of initiating investment effort. With this
specification, the optimal intensity is unique and given by
\begin{equation}
  \label{eq:optimal_supply_intensity}
  \lambda^*(t,s,x)  =
  \bigl(\Delta_sv(t,s,x)-\rho\bigr)_+^{1/(\beta-1)}.
\end{equation}
Investment is positive only when the incremental value of another
capacity addition exceeds the threshold \(\rho\). The curvature
parameter \(\beta\) governs how strongly the optimal investment
intensity responds to that incremental value.

\subsubsection{Numerical results}
\label{sec:single_results}

For the single-technology model, we solve the HJB equation
\eqref{eq:hjb_single} backward in time on a discrete grid for
available supply and reference data-center demand. At each time step,
the investment intensity is updated from the current marginal value of
an additional capacity increment. Holding this policy fixed, a
semi-implicit Euler step produces a sparse linear system for the value
function. Full discretization and implementation details are provided
in Appendix~\ref{app:numerics}.

We now present the numerical results of the model using the parameters
reported in Table~\ref{tbl:parameters_single_numerical}. Price
elasticities are determined using the demand response functions in
Section \ref{sec:demand_response_functions}. The single-technology
results address when investment is attractive and how revenue
cannibalization can limit further investment.
\begin{table}[htb]
  \centering
  \begin{tabular}{|l|l|}
    \hline
    Parameter & Value \\
    \hline
    Generator size \((\delta)\) & 100 MW \\
    \hline
    Reference traditional demand growth rate \((\gamma)\) & 3\% \\
    \hline
    Reference data-center demand increment (\(\kappa\)) & 225 MW \\
    \hline
    Data center intensity \((\mu)\) & \(6 / 0.225\approx 26\) per year \\
    \hline
    Cost power \((\beta)\) & \(2\) \\
    \hline
    Cost linear \((\rho)\) & \(0\) \\
    \hline
    Discount rate \((r)\) & 3\% \\
    \hline
  \end{tabular}
  \caption{Parameters in the single supply numerical simulations.}
  \label{tbl:parameters_single_numerical}
\end{table}

Figure~\ref{fig:single_value_and_policy} shows how the value function
and the corresponding optimal investment policy vary with available
supply and reference data-center demand.
\begin{figure}[H]
  \centering
  \begin{subfigure}[t]{0.48\textwidth}
    \centering
    \includegraphics[width=\textwidth]{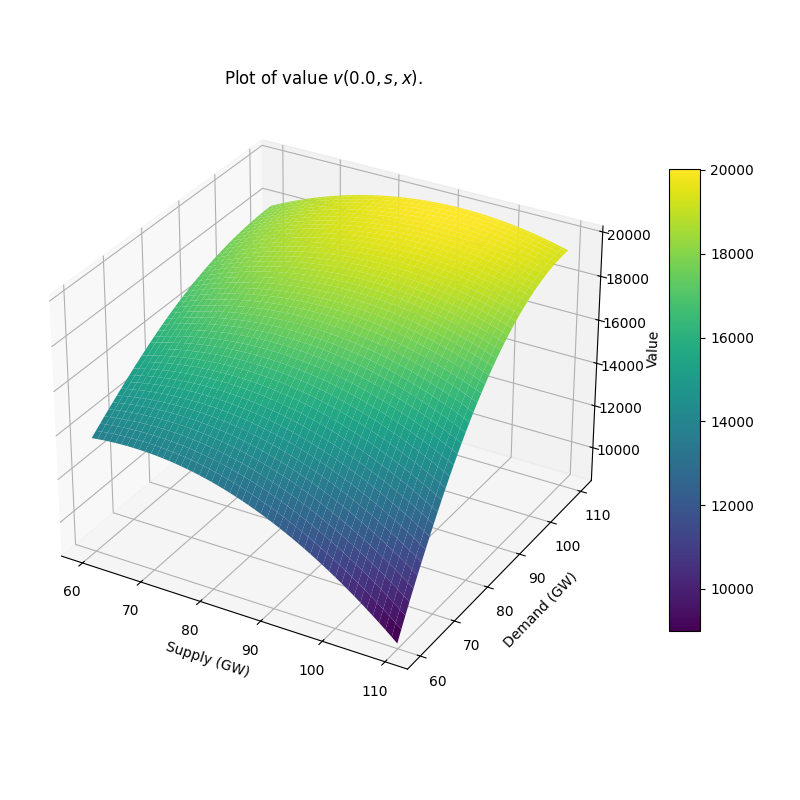}
    \caption{Value function \(v(0,s,x)\).}
    \label{fig:value_surface}
  \end{subfigure}
  \hfill
  \begin{subfigure}[t]{0.48\textwidth}
    \centering
    \includegraphics[width=\textwidth]{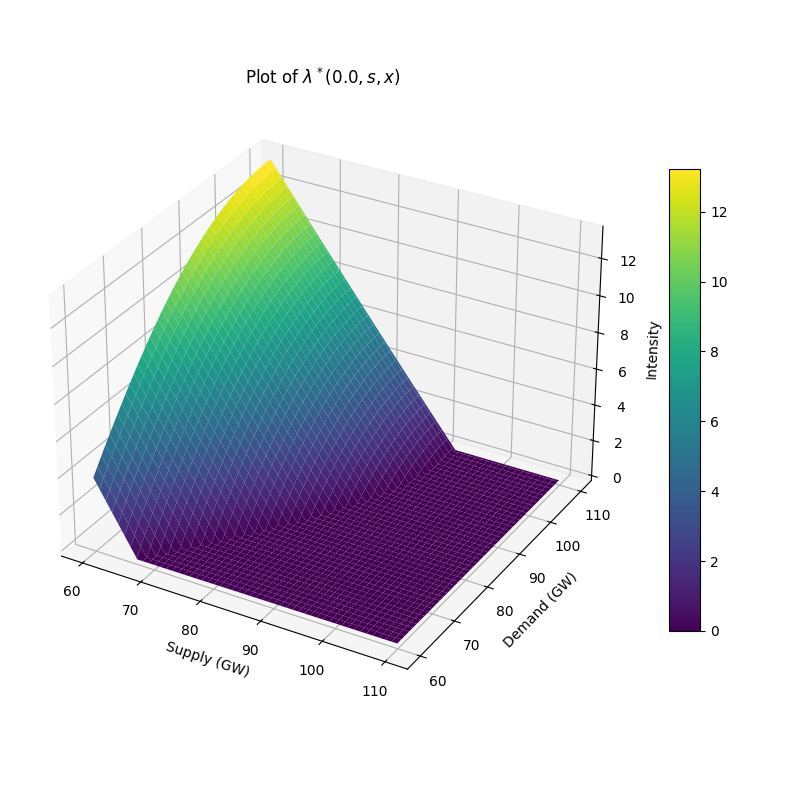}
    \caption{Optimal intensity \(\lambda^*(0,s,x)\).}
    \label{fig:lambda_surface}
  \end{subfigure}
  \vspace{0.4cm}
  \begin{subfigure}[t]{0.48\textwidth}
    \centering
    \includegraphics[width=\textwidth]{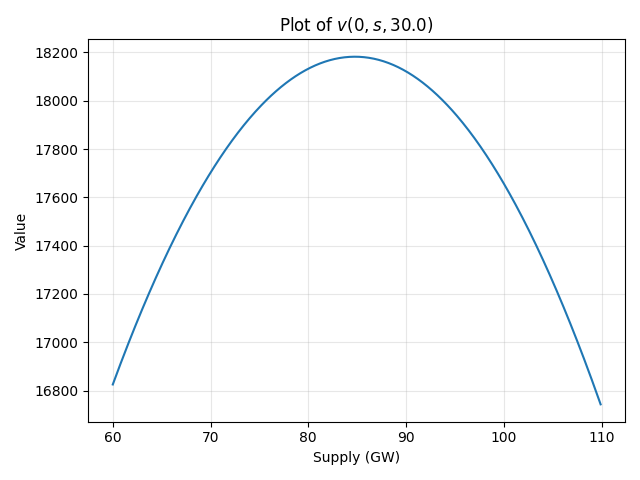}
    \caption{Cross-section \(v(0,s,30)\).}
    \label{fig:value_cross_section}
  \end{subfigure}
  \hfill
  \begin{subfigure}[t]{0.48\textwidth}
    \centering
    \includegraphics[width=\textwidth]{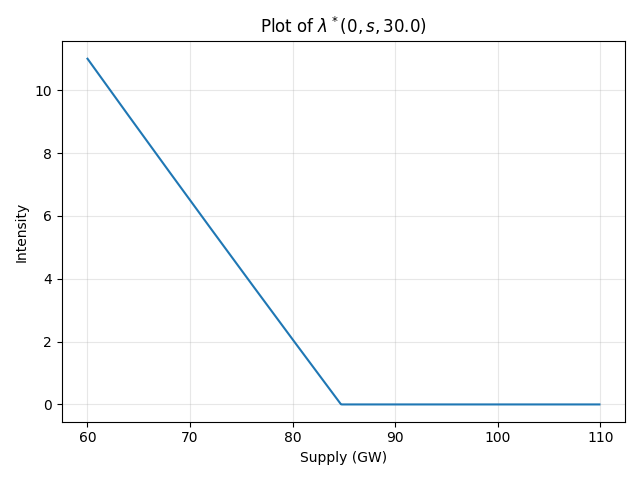}
    \caption{Cross-section \(\lambda^*(0,s,30)\).}
    \label{fig:lambda_cross_section}
  \end{subfigure}
  \caption{Single-technology numerical results at the initial time
    using the calibration in
    Table~\ref{tbl:parameters_single_numerical}.
    Panels~(\subref{fig:value_surface})
    and~(\subref{fig:lambda_surface}) show the value function and
    corresponding optimal investment intensity over the state space,
    while Panels~(\subref{fig:value_cross_section})
    and~(\subref{fig:lambda_cross_section}) show the corresponding
    cross-sections at \(x=30\). In the surface plots, the demand is
    the reference demand from traditional consumers plus data
    centers.}
  \label{fig:single_value_and_policy}
\end{figure}
Panels~(\subref{fig:value_surface})
and~(\subref{fig:value_cross_section}) show that the value function
increases with reference data-center demand, which reflects the
greater revenue opportunity created by additional electricity demand.
As a function of supply, value increases up to a point and then
decreases. Panels~(\subref{fig:lambda_surface})
and~(\subref{fig:lambda_cross_section}) show the optimal investment
intensity from formula \eqref{eq:optimal_supply_intensity}. Investment
is concentrated in states with relatively low available supply and
sufficiently high reference data-center demand. In these states,
scarcity keeps the marginal value of new capacity above the cost
threshold \(\rho\). Once available supply rises beyond the region in
which another increment has sufficient value, the optimal intensity
falls to zero. Thus a high portfolio value does not by itself imply
continued investment; investment depends on the incremental value of
the next capacity addition.

\subsubsection{Monte Carlo simulations}
\label{sec:single_simulation}

The Monte Carlo simulations summarized in
Figure~\ref{fig:single_supply_mc_choke} illustrate how demand growth,
capacity additions, and investment incentives interact under the
optimal policy.
\begin{figure}[htbp]
  \centering
  \begin{subfigure}{0.48\textwidth}
    \centering
    \includegraphics[width=\textwidth]{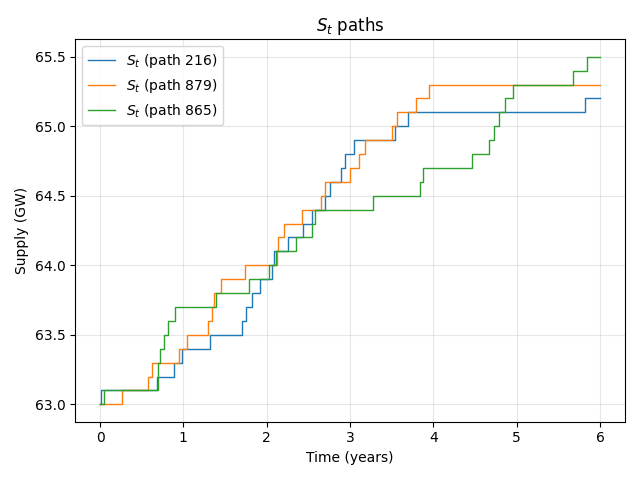}
    \caption{}
    \label{fig:revenue_cannibalization_a}
  \end{subfigure}
  \begin{subfigure}{0.48\textwidth}
    \centering
    \includegraphics[width=\textwidth]{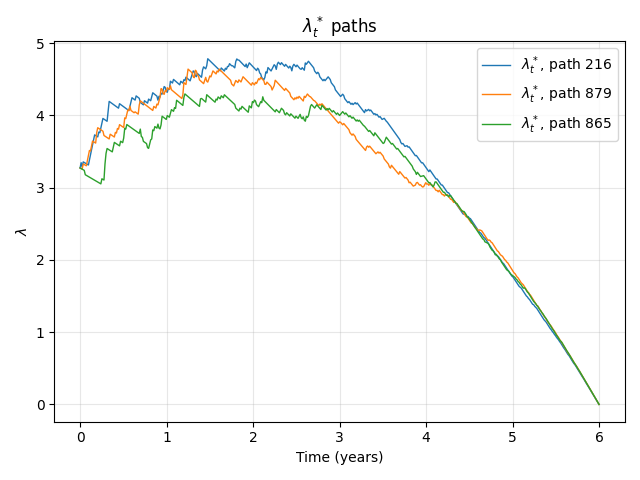}
    \caption{}
    \label{fig:revenue_cannibalization_b}
  \end{subfigure}
  \begin{subfigure}{0.48\textwidth}
    \centering
    \includegraphics[width=\textwidth]{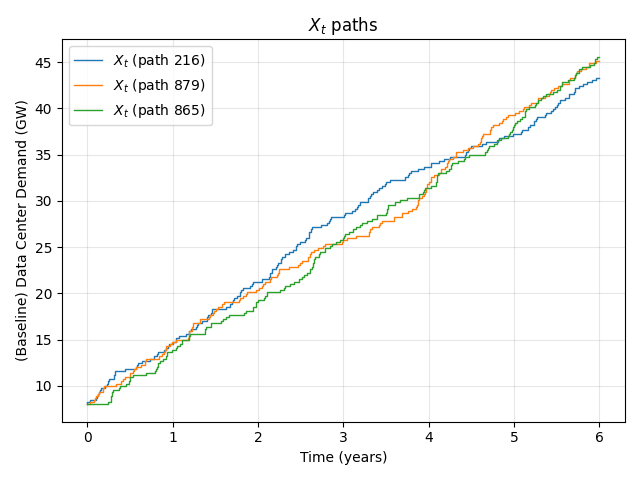}
    \caption{}
    \label{fig:revenue_cannibalization_c}
  \end{subfigure}
  \begin{subfigure}{0.48\textwidth}
    \centering
    \includegraphics[width=\textwidth]{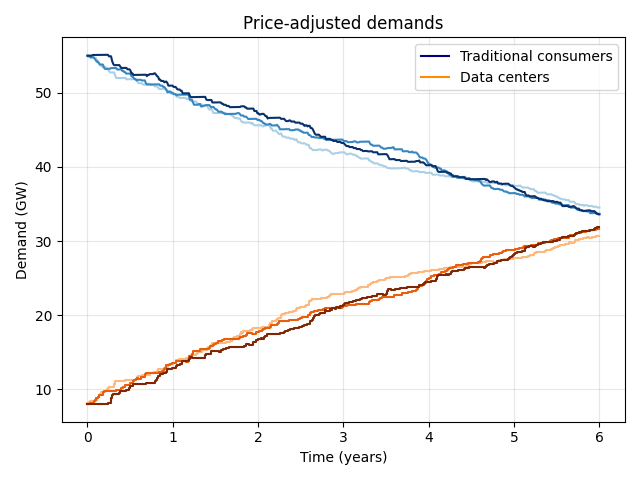}
    \caption{}
    \label{fig:revenue_cannibalization_d}
  \end{subfigure}
  \begin{subfigure}{0.48\textwidth}
    \centering
    \includegraphics[width=\textwidth]{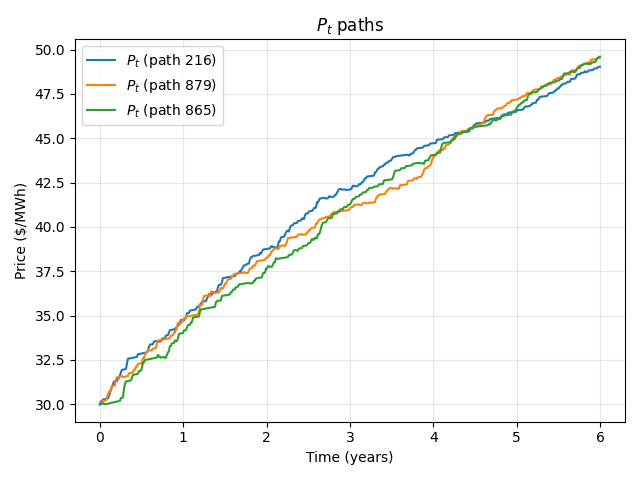}
    \caption{}
    \label{fig:revenue_cannibalization_e}
  \end{subfigure}
  \begin{subfigure}{0.48\textwidth}
    \centering
    \includegraphics[width=\textwidth]{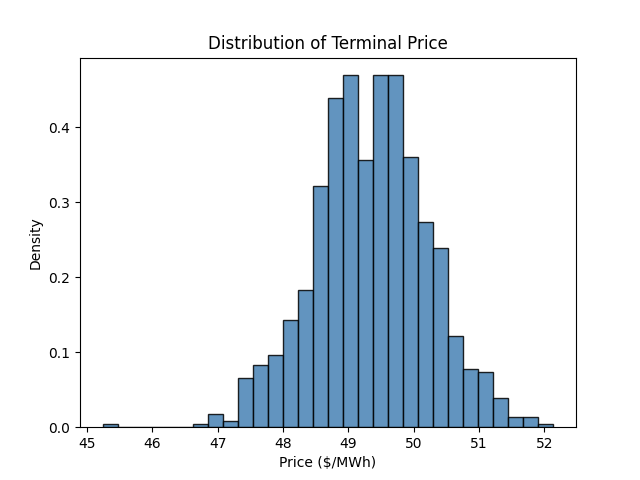}
    \caption{}
    \label{fig:revenue_cannibalization_f}
  \end{subfigure}
  \caption{Summary of single-technology Monte Carlo simulation
    \((N=1000)\) under the parameters in
    Table~\ref{tbl:parameters_single_numerical}. Top: (a) three
    simulated paths of total generation capacity \(S_t\), and (b) the
    corresponding optimal intensity paths \(\lambda_t^*\). Middle: (c)
    three simulated paths of reference data-center demand \(X_t\), and
    (d) price-responsive demand of traditional consumers and data
    centers. Bottom: (e) three simulated price paths, and (f)
    simulated terminal-price distribution.}
  \label{fig:single_supply_mc_choke}
\end{figure}
Panel~(\subref{fig:revenue_cannibalization_a}) shows three controlled
paths of available generation capacity, with discrete increases
corresponding to build-out completions. The associated optimal
investment intensities are shown in
panel~(\subref{fig:revenue_cannibalization_b}), which demonstrates the
combined effects of capacity accumulation and the shortening remaining
investment horizon.

On the demand side, panel~(\subref{fig:revenue_cannibalization_c})
shows that reference data-center demand continues to increase through
repeated exogenous arrivals. Thus, the decline in optimal investment
intensity can occur even while reference data-center demand continues
to grow. Panel~(\subref{fig:revenue_cannibalization_d}) shows the
corresponding price-responsive demands of traditional consumers and
data centers, while panel~(\subref{fig:revenue_cannibalization_e})
shows the resulting market-clearing price paths.

We highlight the opposing effects of rising reference data-center
demand and additional generation capacity: demand arrivals place
upward pressure on the market-clearing price while completed supply
investments reduce scarcity and place downward pressure on it. The
terminal-price distribution across the full Monte Carlo sample is
shown in panel~(\subref{fig:revenue_cannibalization_f}). The observed
dispersion in terminal prices in
panel~(\subref{fig:revenue_cannibalization_f}) illustrates optimal
investment does not eliminate price uncertainty since the timing of
demand arrivals and project completions remains stochastic.

More importantly, the decline of the optimal intensity in
panel~(\subref{fig:revenue_cannibalization_b}) despite continued
growth in panel~(\subref{fig:revenue_cannibalization_c}) illustrates
the revenue-cannibalization mechanism. Adding capacity supports
additional sales but also lowers the market-clearing price earned on
the owner's existing generation portfolio, eventually reducing the
incremental value of further investment.

\subsubsection{Explicit Formula Example}
\label{sec:no_investment}

As we have seen, higher reference data-center demand raises the
market-clearing price and generally increases the value of additional
capacity, while greater available supply lowers the price. The
investment incentive is therefore governed by a tradeoff: additional
capacity creates revenue from new output, but also depresses the
market-clearing price earned on the producer's existing capacity. In
sufficiently well-supplied states, this revenue-cannibalization effect
can make the incremental value of further capacity nonpositive even
while reference data-center demand continues to grow. We give here an extreme analytical example where there is complete
cannibalization.

We suppose a slightly different type of demand model in which price
response {\em additively} changes the hyperscaler and traditional
demands $I$ and $X$ at the reference price, rather than
multiplicatively as in \eqref{eq:demand_general}. Specifically we take
$I_t\equiv I_0$ ($\gamma=0$) and
\[
  D_1(I_t,P_t) = I_0-\frac12\alpha\log(P_t/P_0),\qquad D_2(X_t,P_t) =
  X_t-\frac12\alpha\log(P_t/P_0),
\]
where $\alpha>0$ is a conversion parameter in GW. In this model,
demand is equated to risk-adjusted supply given by
$S_t+\alpha\log(S_t/S_b)$, where $S_b>0$ is ``large" so the
reliability adjustment is negative for $S_t<S_b$. As supply capacity
increases, supply is treated as more reliable and the modification is
smaller.

Then the demand-supply market clearing condition
\[
  I_0-\frac12\alpha\log(P_t/P_0) + X_t-\frac12\alpha\log(P_t/P_0) =
  S_t+\alpha\log(S_t/S_b)
\]
gives $P_t=ke^{(X_t-S_t)/\alpha}/S_t$, where $k$ collects the
constants $(I_0, P_0, S_b)$. We assume for simplicity $\alpha=k=1$.
Therefore we have the revenue at time $t$ is given by
$S_tP_t = e^{X_t-S_t}$. It is exponentially increasing in hyperscaler
demand, and exponentially decreasing in installed capacity.

We further suppose that cost of supply investment depends on $X_t$ and
$S_t$ as well in the same functional form:
\[
  C(\lambda)\mapsto C(\lambda; X_t,S_t) = e^{X_t-S_t} c(\lambda),
\]
where \(c\) is increasing, strictly convex and $c(0)=0$. So the cost
of building new supply is high when hyperscaler demand is high because
everyone is trying to get in, but it is low when supply is already
high. We also assume $c'(0)=0$.

Then \eqref{eq:hjb_single} becomes
\begin{equation}
  \label{eq:hjb_exponential_price}
  \partial_tv(t,s,x)+ \sup_{\lambda \geq 0} \left[ \lambda \Delta_s v(t,s,x) - e^{x-s} c(\lambda) \right ] + \mu \Delta_x v(t,s,x) + e^{x-s} - rv(t,s,x)=0.
\end{equation}
We look for a solution of the form
\begin{equation}
  v(t, s, x) = e^{x-s} g(t),\label{ansatz}
\end{equation}
for some function \(g\) to be found. Note that we have
\begin{align*}
  \Delta_s v(t,s,x) &= v(t,s+\delta,x) - v(t,s,x) = e^{x-s} g(t) (e^{-\delta}-1), \\
  \Delta_x v(t,s,x) &= v(t,s,x+\kappa) - v(t,s,x)= e^{x-s} g(t) (e^{\kappa}-1).
\end{align*}
The optimization problem is
\[
  e^{(x-s)} \sup_{\lambda\geq 0} \left[ \lambda g(t)( e^{-\delta}-1 )
    - c(\lambda) \right],
\]
which gives
\begin{equation}
  \label{eq:lam-opt_exponential_price}
  \lambda_t^* = \begin{cases}
    0, & g(t)\bigl(e^{-\delta}-1\bigr)\le 0, \\[4pt]
    (c')^{-1}\!\left(g(t)\bigl(e^{-\delta}-1\bigr)\right), & g(t)\bigl(e^{-\delta}-1\bigr)>0.
\end{cases}
\end{equation}
Substituting the \textit{ansatz} \eqref{ansatz} into
\eqref{eq:hjb_exponential_price} gives the nonlinear ODE
\begin{equation}
  \label{eq:g_nonlinear_ode}
  g'(t) +  g(t)\lambda^*(t) (e^{-\delta}-1) - c(\lambda^*(t))+ \mu
  g(t)(e^{\kappa}-1)+1 - rg(t)=0, \qquad
  g(T) = 0.
\end{equation}

From \eqref{eq:lam-opt_exponential_price}, in order that
$\lambda^*_t>0$, we must have $g(t)<0$. However, since $g(T^-)=0$, we
have $\lambda^*_{T^-}=0$ and therefore $g'(T^-)=-1$, which means $g$
is decreasing to zero as $t\uparrow T$. If at some earlier time
\(t_0<T\) we had \(g(t_0)=0\), then \(\lambda^*_{t_0}=0\) and the same
argument gives $g'(t_0^-)=-1$. Hence \(g(t)>0\) for \(t<t_0\)
sufficiently close to \(t_0\), so \(g\) cannot cross from nonnegative
values into negative values as the equation is solved backward from
\(T\). Consequently, $g(t)\geq0$ and $\lambda^*_t=0$ for all $t$.

This example illustrates revenue cannibalization: additional supply
depresses the market-clearing price sufficiently that a
revenue-maximizing generation owner optimally chooses never to invest,
even when demand is growing.

\subsection{Multiple Supply Technologies}
\label{sec:multiple_technologies}

We now extend the single-technology framework by allowing the investor
to allocate investment effort across \(d\) generation technologies.
For each technology \(j=1,\ldots,d\), \(\delta_j\) is the capacity
delivered by one completed project, \(\lambda_{j,t}\geq0\) is the
controlled completion intensity, and \(C_j\) is the
technology-specific investment cost. The aggregate controlled supply
process is
\[
  dS_t = \sum_{j=1}^d \delta_j\,dN_t^{\lambda_j}.
\]
The value function becomes
\[
  v(t,s,x) = \sup_{\boldsymbol{\lambda}\in\Lambda}
  \mathbb{E}\left[\int_t^T e^{-r(u-t)} \left(S_uP(u,S_u,X_u) -
      \sum_{j=1}^d C_j(\lambda_{j,u}) \right)\,du \;\middle|\; S_t=s,\
    X_t=x \right].
\]
The multi-technology HJB is
\begin{equation}
  \label{eq:multi_control_hjb}
  \partial_t v(t,s,x) + \sum_{j=1}^d \sup_{\lambda_j\ge 0} \left\{\lambda_j\Delta_{s_j}v(t,s,x)-C_j(\lambda_j) \right\} + \mu\Delta_xv(t,s,x) + sP(t,s,x) - rv(t,s,x) =0,
\end{equation}
with \(v(T,s,x)=0\), where
\[
  \Delta_{s_j}v(t,s,x) = v(t,s+\delta_j,x)-v(t,s,x), \quad\mbox{and
  }\quad \Delta_xv(t,s,x) = v(t,s,x+\kappa)-v(t,s,x).
\]
A multivariate derivation and verification argument are given in
Appendix~\ref{app:hjb}.

For the power-cost specification
\[
  C_j(\lambda_j) = \frac{1}{\beta_j}\lambda_j^{\beta_j} +
  \rho_j\lambda_j,
\]
the optimal intensity for technology \(j\) is
\[
  \lambda_j^*(t,s,x) = \left(\Delta_{s_j}v(t,s,x)-\rho_j
  \right)_+^{1/(\beta_j-1)}.
\]
Investment in technology \(j\) is positive only when the incremental
value of another completed asset exceeds the threshold \(\rho_j\).
Technologies differ both in the size \(\delta_j\) of a completed
build-out and in the cost of raising its completion intensity. The
investor therefore allocates effort by comparing the marginal value of
each technology-specific capacity addition with the corresponding
marginal investment cost.

Because the horizon is finite and future revenues are discounted,
technologies whose expected benefits arrive too late relative to the
remaining horizon may receive little or no investment effort.

\subsubsection{Numerical results}
\label{sec:multi_results}

The multi-technology HJB is solved using the same backward
semi-implicit scheme as in the single-technology case, except that the
policy update is performed separately for each \(\lambda_j\). Full
implementation details are provided in Appendix~\ref{app:numerics}.

\paragraph{Technology choice.}
We run experiments in the case when there are multiple supply
technologies. The parameters used in these experiments are presented
in Table~\ref{tbl:multi_parameters}. The multi-technology results ask
how the investor allocates effort when technologies differ in capacity
increments and cost parameters.

Figure~\ref{fig:multi_lam-opt_numerical_choke} reports the optimal
investment intensity for each generation technology as a function of
available supply and reference data-center demand at the initial time.
The panels therefore show how the same market state can lead to
different technology-specific project-completion intensities after
solving \eqref{eq:multi_control_hjb}.

\begin{table}[htbp]
  \centering
  \begin{tabular}{|l|l|}
    \hline
    Parameters & Value \\
    \hline
    Completed-project capacity increments
    \((\boldsymbol{\delta})\) &
    \((250,500,50,100,1000,250)\ \mathrm{MW}\) \\
    \hline
    Reference traditional demand growth rate \((\gamma)\) & 3\% \\
    \hline
    Reference data-center demand increment \((\kappa)\) & 225 MW \\
    \hline
    Cost power \((\boldsymbol{\beta})\) & \((2,2,2,2,2,2)\) \\
    \hline
    Cost linear \((\boldsymbol{\rho})\) & \((4, 25, 0, 0, 50, 10)\) \\
    \hline
    Discount rate \((r)\) & 3\% \\
    \hline
  \end{tabular}
  \caption{Parameters used in the multiple supply technologies
  numerical experiments.}
  \label{tbl:multi_parameters}
\end{table}

The entries of \(\boldsymbol{\delta}\), \(\boldsymbol{\beta}\), and
\(\boldsymbol{\rho}\) are ordered as natural gas, coal, solar, wind,
large-scale nuclear, and SMR nuclear. The technology-specific cost
parameters are reduced-form, illustrative parameters that jointly
represent capital requirements, development difficulty, and the effort
required to increase the expected project-completion rate.
\begin{figure}[htbp]
  \centering
  \includegraphics[width=0.65\textwidth]
  {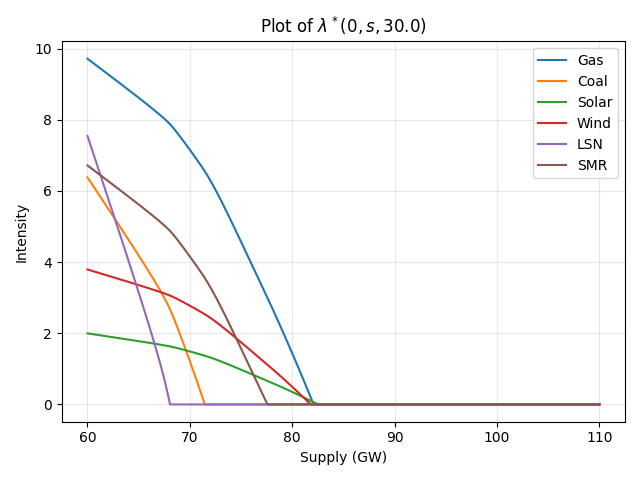}
  \caption{Optimal investment intensities across generation technologies
  as a function of available supply, holding
  reference data-center demand fixed, using the
  parameters in Table~\ref{tbl:multi_parameters}.}
  \label{fig:multi_lam-opt_numerical_choke}
\end{figure}

Figure~\ref{fig:multi_lam-opt_numerical_choke} compares the optimal
intensities across technologies at a common level of reference
data-center demand. The comparison shows that technology choice is
state dependent. At low available supply, the optimal intensity
incorporates all technologies but prioritizes high capacity projects,
namely natural gas, both types of nuclear, and coal, to try to catch
up to demand quickly. As available supply ramps up, large-scale
nuclear and coal, the highest cost but highest capacity technologies,
fall to no investment while relatively inexpensive natural gas, SMR
nuclear, solar, and wind all continue to attract investment until
available supply meets demand. High cost projects have rapidly
decreasing marginal reward as the deficit between available supply and
demand narrows.

These differences reflect both the technology-specific capacity
increments \(\delta_j\) and the corresponding investment-cost
parameters in Table~\ref{tbl:multi_parameters}. In particular, a
larger optimal project-completion intensity does not necessarily imply
a larger rate of capacity addition, since each completed project
contributes a different amount \(\delta_j\) of generation capacity.
Thus, Figure~\ref{fig:multi_lam-opt_numerical_choke} compares optimal
investment effort across technologies rather than simply comparing
nameplate capacities or levelized generation costs.

\paragraph{Controlled investment paths}
We conclude the section with Monte Carlo simulations in the
multi-technology case, summarized in
Figure~\ref{fig:multi_supply_mc_choke}.
Panel~(\subref{fig:multi_mc_a}) shows total available
generation-capacity paths, and panel~(\subref{fig:multi_mc_b}) shows
the technology-specific optimal investment intensities. The exogenous
reference data-center demand paths are shown in
Panel~(\subref{fig:multi_mc_c}), while panel~(\subref{fig:multi_mc_d})
reports the corresponding price-responsive demands. The
market-clearing price paths are shown in
panel~(\subref{fig:multi_mc_e}), and panel~(\subref{fig:multi_mc_f})
gives the terminal-price distribution.

\begin{figure}[htbp]
  \centering
	  \begin{subfigure}{0.48\textwidth}
	    \includegraphics[width=\textwidth]{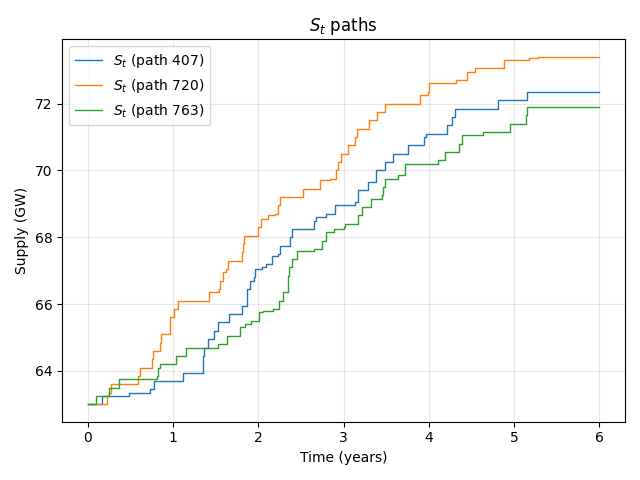}
	    \caption{}\label{fig:multi_mc_a}
	  \end{subfigure}
	  \begin{subfigure}{0.48\textwidth}
	    \includegraphics[width=\textwidth]{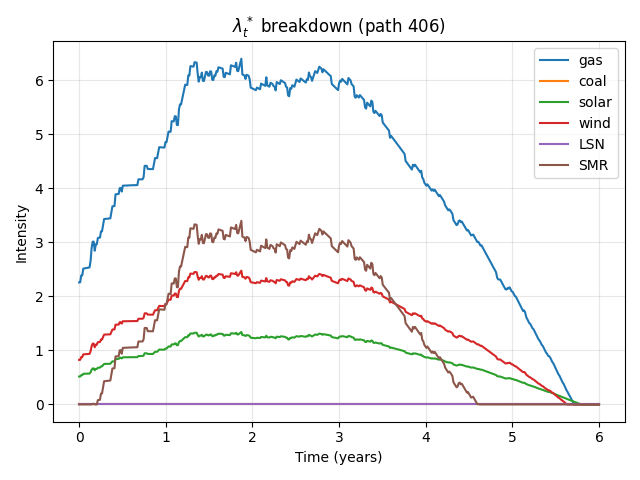}
	    \caption{}\label{fig:multi_mc_b}
	  \end{subfigure}
	  \begin{subfigure}{0.48\textwidth}
	    \includegraphics[width=\textwidth]{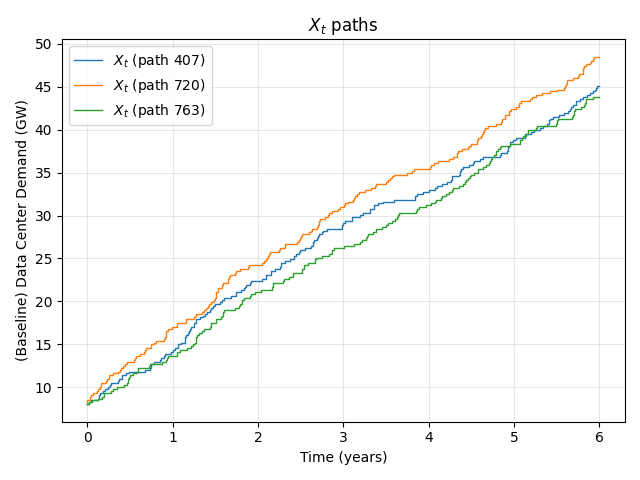}
	    \caption{}\label{fig:multi_mc_c}
	  \end{subfigure}
	  \begin{subfigure}{0.48\textwidth}
	    \includegraphics[width=\textwidth]{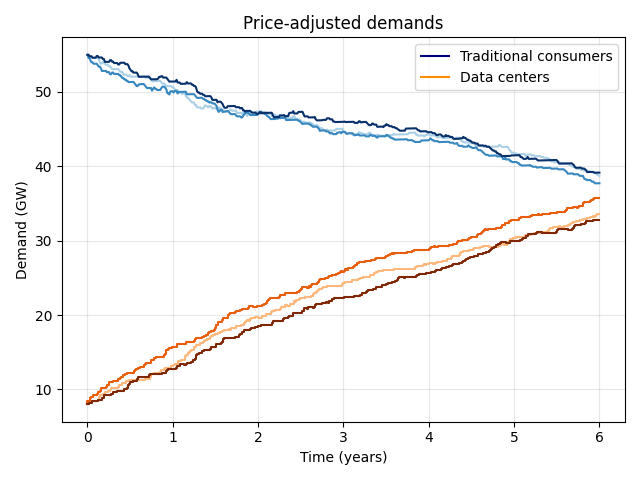}
	    \caption{}\label{fig:multi_mc_d}
	  \end{subfigure}
	  \begin{subfigure}{0.48\textwidth}
	    \includegraphics[width=\textwidth]{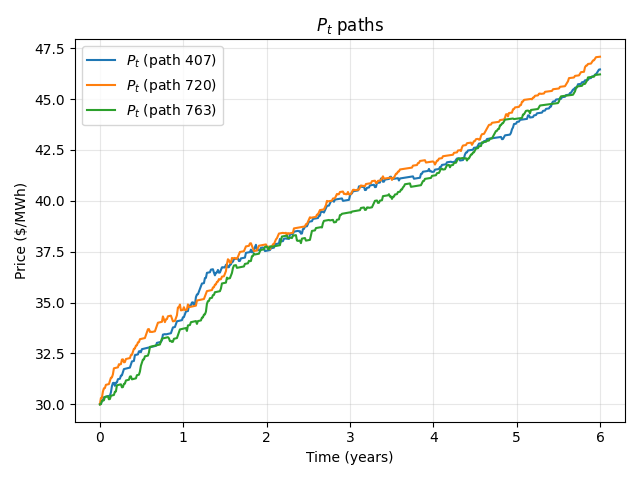}
	    \caption{}\label{fig:multi_mc_e}
	  \end{subfigure}
	  \begin{subfigure}{0.48\textwidth}
	    \includegraphics[width=\textwidth]{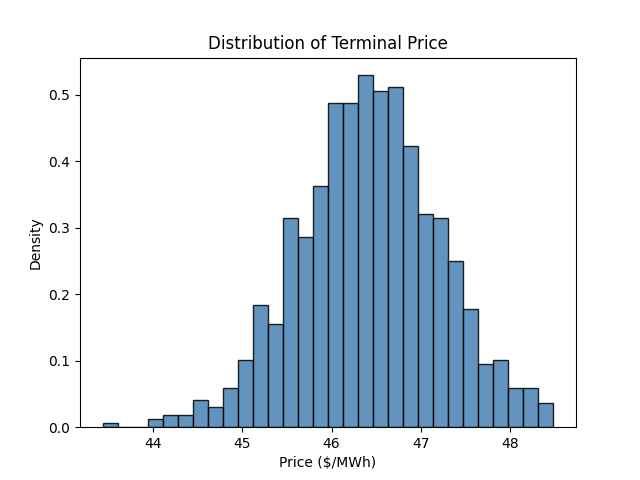}
	    \caption{}\label{fig:multi_mc_f}
	  \end{subfigure}
      \caption{Summary of multi-technology Monte Carlo simulation
        \((N = 1000)\) under the parameters in
        Table~\ref{tbl:multi_parameters}. Top: (a) three simulated
        paths of total generation capacity \(S_t\), and (b) intensity
        mix breakdown \(\boldsymbol{\lambda}^*_t\) of one simulated
        \(S_t\) path. Middle: (c) three simulated paths of reference
        data-center demand \(X_t\), and (d) price-responsive demand of
        traditional consumers and data centers. Bottom: (e) three
        simulated price paths, and (f) simulated terminal-price
        distribution.}
  \label{fig:multi_supply_mc_choke}
\end{figure}

The simulated paths show how the investor reallocates effort across
technologies as the state evolves. Increases in
reference data-center demand raise the value of
additional capacity, while completed capacity additions increase supply
and reduce the marginal value of later additions. Technologies respond
differently because they differ in project size and investment cost.
Demand arrivals put upward pressure on price, while capacity
completions put downward pressure on price. Dynamically reallocating
investment across technologies does not eliminate terminal-price
uncertainty: the timing of demand arrivals and project completions still
generates a range of possible market outcomes. The main difference from
the single-technology model is the additional technology-choice margin,
 \section{Conclusion}
\label{sec:conclusion}

We have developed a framework linking data-center load growth,
generation capacity, and market-clearing electricity prices. Absent a
supply response, data-center demand growing at ERCOT's projected pace
roughly doubles the average wholesale price within a decade, and the
discrete, uncertain arrival of load and capacity turns any point
forecast into a wide distribution of outcomes. Endogenizing investment
then shows that because each completed project lowers the price earned
on the investor's entire portfolio, optimal investment intensity falls
toward zero as capacity accumulates even while data-center load
continues to arrive. Under our calibrations, prices rise from \$30/MWh
to roughly \$46–49/MWh over six years despite optimal investment.

These conclusions should be read in contrast to recent publications
arguing that new data-center loads can, and in some cases have,
lowered costs to residential customers, particularly where supply is
already available. A May 2026 white paper by the Energy Systems
Integration Group and the Brattle Group explains that utilities with
unused capacity can lower rates because the fixed costs of the grid
are spread over a larger demand base \cite{esig2026ratelargeLoads}. An
October 2025 regression study by Lawrence Berkeley National Laboratory
and the Brattle Group found that, from 2019 to 2024, states with the
highest load growth ``experienced reductions in real prices,'' whereas
states with contracting loads generally saw prices rise
\cite{wiser2025factors}. Using an instrumental-variables approach, an
EPRI and Watershed working paper estimates that data centers modestly
reduced average residential rates between 2015 and 2024 through the
same fixed-cost-spreading mechanism \cite{watten2026datacenters}.

The past decade, however, may not be a reliable guide to the effects
of new large loads in the future. This evidence is retrospective and
estimated over a period in which aggregate load growth was modest by
the standard of current forecasts and existing capacity was
underutilized. It says little about wholesale energy and capacity
prices, which are set by supply–demand clearing, passing a higher cost
to every customer once supply is tight. The EPRI and Watershed authors
themselves caution that emerging supply constraints could reverse the
effect they identify \cite{watten2026datacenters}.

With forecasts of up to 40\% load growth over the next decade, and
considerable uncertainty about whether and when it will materialize,
many regions of the United States may enter a period of severe supply
tightness. Building new generation at the required pace could run into
supply-chain limits and rising incremental costs that push wholesale
energy and capacity prices higher, costs ultimately passed on to
households and existing businesses. Our results add an economic reason
for the supply response to lag even where physical constraints do not
bind: an investor who expands capacity erodes the scarcity rents that
motivated the investment. In the regime ahead, rising wholesale
prices, not the dilution of fixed costs, will be the stronger factor
in what existing customers pay, and capturing this requires a
structural model of price formation because that regime lies outside
the historical sample on which the retrospective studies rest.

Several extensions would sharpen these conclusions: endogenous
data-center demand, in which developers choose when to connect, how
much capacity to request, and whether to sign long-term contracts; a
distinction between project initiation and completion to represent
construction pipelines, cancellation risk, and uncertain
time-to-build; a retail-rate layer that weighs the dilution effect
directly against the wholesale effect analyzed here.
 \appendix
\section{Formulas \& Figures for Section \ref{sec:no_new}}
\label{app:fixed_supply}

In Section \ref{sec:no_new}, when $X$ grows linearly, the
market-clearing price is now determined by
\[
  I_tF_1(P_t)+X_tF_2(P_t)=S_0.
\]
As long as both groups remain active, we have for $ P_t<A_1$:
\begin{equation}
  \label{eq:Pt_formula}
  P_t = A_2+\frac{R_t}{2} -\frac{1}{2} \sqrt{R_t^2 +4(A_2-A_1)R_t +4(A_2-P_0)^2\frac{S_0}{X_t}}, \quad R_t \coloneqq \frac{(A_2-P_0)^2}{A_1-P_0}\frac{I_t}{X_t}.
\end{equation}

Define \(\tau\) to be the time at which traditional consumers reach
their effective choke price, i.e. \(P_\tau=A_1\). Since \(A_1<A_2\),
we have \(F_2(A_1)>0\). Moreover, \(X_t=X_0+c_Xt\) grows without
bound, while supply remains fixed at \(S_0\). Hence \(X_tF_2(A_1)\)
eventually reaches \(S_0\), so \(\tau<\infty\), unless group 1 is
already priced out at \(t=0\). At \(t=\tau\), price-responsive
traditional demand is zero:
\[
  D_1(I_\tau,A_1) = I_\tau F_1(A_1) = 0.
\]
The market-clearing condition therefore becomes $S_0=X_\tau F_2(A_1)$,
which, using the definition of \(F_2\), gives that the dropout time,
when traditional demand is gone, is given by
\begin{equation}
  \label{eq:kdef}
  \tau = \frac{1}{c_X} \left(kS_0 -X_0 \right)^+, \quad k\coloneqq \left(\frac{A_2-P_0}{A_2-A_1} \right)^2.
\end{equation}
Then $\tau>0$ as long as initial reference data center demand
$X_0<kS_0$; otherwise group 1 is already priced out at the initial
time and one sets \(\tau=0\).

For \(t\geq\tau\), group 1 is no longer active, and market clearing is
determined entirely by price-responsive data-center demand:
$S_0=X_tF_2(P_t)$. It follows that
\[
  P_t = A_2-(A_2-P_0) \sqrt{\frac{S_0}{X_0+c_Xt}}, \qquad t \geq \tau,
\]
and so $ \lim_{t\to\infty}P_t=A_2$. Thus the price crosses the
traditional-consumer choke price $A_1$ in finite time, but approaches
the hyperscaler choke price $A_2$ only asymptotically. In this sense,
sufficiently rapid growth in reference data-center demand can price
the more elastic traditional demand out of the market when supply does
not expand. The price nevertheless remains below the hyperscaler
effective choke price \(A_2\) at every finite time.

Figure~\ref{fig:hyperscaler_price_perturbation} illustrates this
predicted price growth over time using parameter values given in
Sections~\ref{sec:market_clearing} and \ref{sec:deterministic_growth},
and Figure~\ref{fig:fixed_supply_demand_decomposition} offers insight
into how each group's demand evolves under the predicted. The
reference data-center demand shifts the price path upward, but it does
not cross \(A_1=\$70/\mathrm{MWh}\) until after approximately
\(\tau=22.3\) years, after which traditional demand is negligible and
the market is supported by price-responsive data-center demand. The
price remains below the hyperscaler effective choke price \(A_2\) at
every finite time.
\begin{figure}[H]
  \centering
  \includegraphics[width=0.7\textwidth]
  {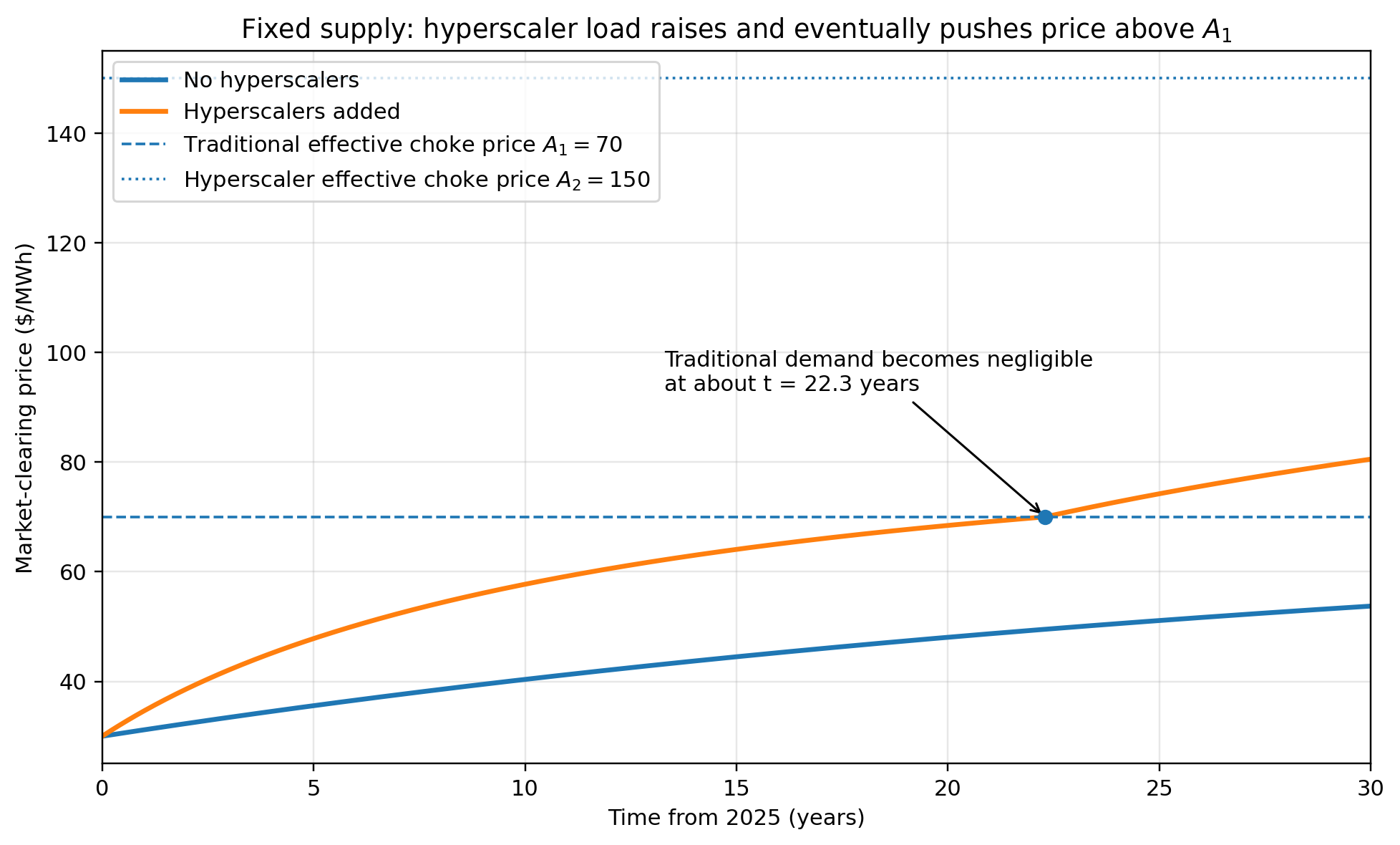}
  \caption{Effect of growth in reference data-center demand on the market clearing price
  under fixed supply.}
  \label{fig:hyperscaler_price_perturbation}
\end{figure}
\begin{figure}[H]
  \centering
  \begin{subfigure}{0.5\textwidth}
    \centering
    \includegraphics[width=\linewidth]{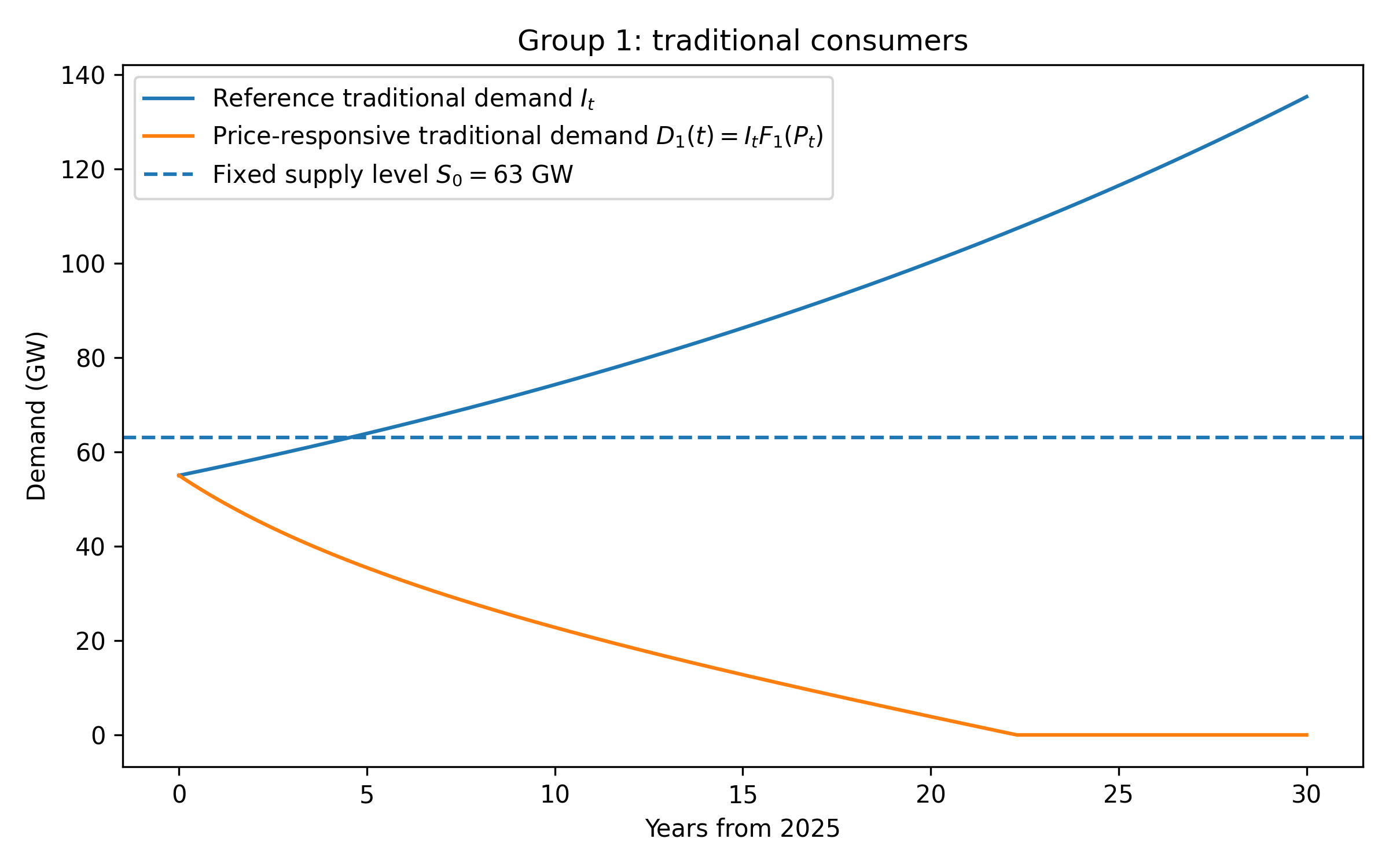}
    \caption{Traditional consumers}
    \label{fig:fixed_supply_traditional_demand}
  \end{subfigure}\hfill
  \begin{subfigure}{0.5\textwidth}
    \centering
    \includegraphics[width=\linewidth]{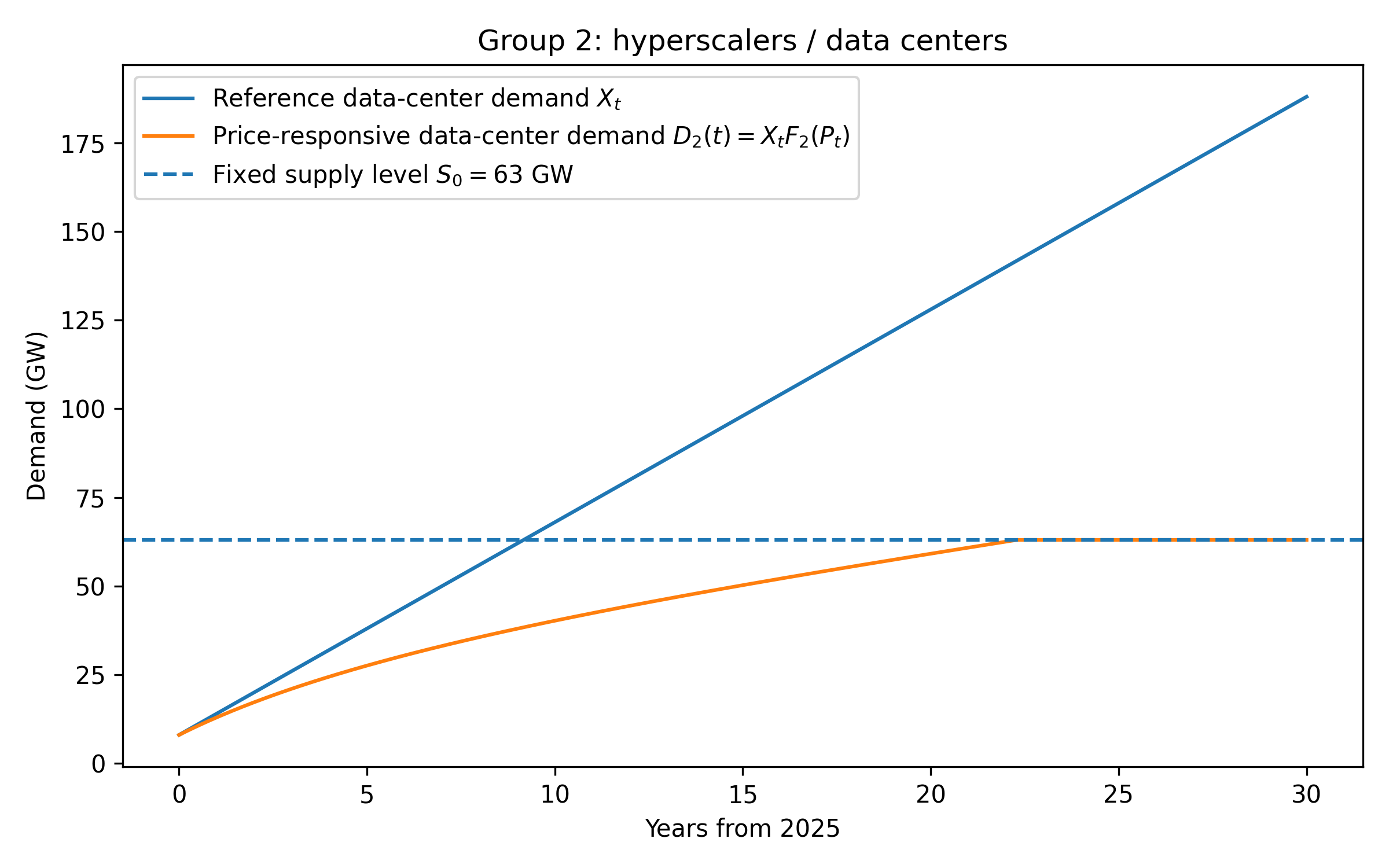}
    \caption{Data centers / hyperscalers}
    \label{fig:fixed_supply_datacenter_demand}
  \end{subfigure}
  \caption{Reference and price-responsive demands under fixed supply.
    Panel~(\subref{fig:fixed_supply_traditional_demand}) shows
    traditional consumers: reference traditional demand \(I_t\) rises
    over time, while price-responsive traditional demand
    \(D_1(t)=I_tF_1(P_t)\) declines as the market-clearing price
    rises. Panel~(\subref{fig:fixed_supply_datacenter_demand}) shows
    data centers: reference data-center demand \(X_t\) rises over
    time, while price-responsive data-center demand
    \(D_2(t)=X_tF_2(P_t)\) increases and eventually absorbs the full
    fixed supply \(S_0\).}
  \label{fig:fixed_supply_demand_decomposition}
\end{figure}

Finally, when supply grows linearly
(Section~\ref{sec:deterministic_supply_growth}), we can re-compute
$\tau$. As long as both groups remain active, the market-clearing
price is obtained from the expression in \eqref{eq:Pt_formula} by
replacing \(S_0\) with \(S_t\). The dropout time of group 1 is again
defined by \(P_\tau=A_1\). Hence, provided \(c_X>kc_S\), we have
$\tau = \frac{kS_0-X_0}{c_X-kc_S}$. This reduces to the fixed-supply
expression \eqref{eq:kdef} when \(c_S=0\). If $c_S\geq \frac{c_X}{k}$,
then supply grows sufficiently rapidly that group 1 never reaches its
effective choke price, and \(\tau=\infty\). With \(c_X=6\)~GW/year as
in Section \ref{sec:deterministic_growth}, the threshold is
$\frac{c_X}{k}\approx 2.67\ \text{GW/year}$. Thus all of the positive
supply-growth rates considered in
Figure~\ref{fig:deterministic_supply_growth} prevent traditional
demand from reaching zero.
 \section{Technical Foundations for the Controlled-intensity Model}
\label{app:hjb}

This appendix gives the dynamic-programming calculation for the
controlled-intensity model used in
Section~\ref{sec:control_intensity}. The formulation follows the
point-process control framework in \cite[Ch. VII, \S
2]{bremaud1981point}, \cite[Chapter 21]{cohen2015stochastic}, and
\cite{hernandezhernandez2019martingaleapproachcontrolgeneral},
specialized to the state variables in the paper: available supply
\(S\), reference data-center demand \(X\), and deterministic reference
traditional demand \(I_t=I_0e^{\gamma t}\).

\subsection{Controlled-intensity Setup}
\label{sec:controlled-intensity_setup}

\paragraph{Stochastic basis.}
Fix a probability space \((\Omega,\mathcal{F},\mathbb{P})\) carrying a
Poisson random measure \(\mathcal{N}(dz,dt)\) on
\(\mathbb{R}_+\times\mathbb{R}_+\) with compensator \(dz\otimes dt\)
and an independent homogeneous Poisson process \(N^\mu\) with
intensity \(\mu>0\). The symbol \(\mu\) is reserved for this exogenous
data-center arrival intensity. Define the compensated random measure
by
\[
  \widetilde{\mathcal N}(dz,dt) = \mathcal N(dz,dt)-dz\,dt.
\]
The filtration \(\mathbb{F}=(\mathcal{F}_t)_{t\geq0}\) is generated by
\(\mathcal{N}\), \(N^\mu\), and the initial conditions, augmented to
satisfy the usual conditions. For a nonnegative predictable intensity
\(\lambda_t\), the controlled supply-arrival process is constructed by
thinning:
\[
  dN_t^\lambda = \int_0^\infty {\bf 1}_{[0,\lambda_t]}(z)\mathcal N(dz,dt) = \lambda_t\,dt + \int_0^\infty {\bf 1}_{[0,\lambda_t]}(z)\widetilde{\mathcal N}(dz,dt).
\]
All expectations below are taken under this fixed probability measure
\(\mathbb{P}\).

\paragraph{State dynamics.}
For the single-technology model, $ dS_t = \delta\,dN_t^\lambda$ and
$dX_t = \kappa\,dN_t^\mu$. The deterministic traditional demand path
is fixed as \(I_t=I_0e^{\gamma t}\). For notation convenience, we wrap
the stochastic state processes into vector
\(Y_t\coloneqq (S_t, X_t)\). The stochastic state has initial
condition \(Y_0 = y_0 = (s_0, x_0)\). Feedback controls use the
predictable left-limit state,
\[
  \lambda_t = \lambda(t, Y_{t-}).
\]
The market-clearing price is written as \(P(t,y)\), with the
deterministic dependence on \(I_t\) absorbed into the explicit time
argument.

\begin{definition}[Admissible controls]
  Let \(p\geq1\) and let \(L\subseteq\mathbb{R}_{\geq 0}\) be the
  action set. For \((t,y)\in[0,T]\times\mathbb{R}_+^2\), the
  admissible class \(\Lambda_p(t,y)\) consists of predictable
  nonnegative controls taking values in \(L\), satisfying
  \[
    \mathbb{E}_{t,y}\left[\int_{t}^T |\lambda_u|^p\,du\right]<\infty,
  \]
  and, when restricted to Markov feedback controls, admitting the form
  \(\lambda_t=\lambda(t,Y_{t-})\) for a measurable function
  \(\lambda: [0, T]\times \mathbb{R}_2^+\to L\). We assume admissible
  controls are stable under concatenation at stopping times.
\end{definition}

The running payoff is producer revenue net of investment cost:
\[
  f(t, y, \lambda) = s P(t, y) - C(\lambda), \qquad
  f:[0,T]\times\mathbb{R}_+^2\times L \to \mathbb{R},
\]
where \(y = (s, x)\). This payoff may be negative because investment
costs can exceed contemporaneous revenue. Using constant discount rate
\(r\), the expected remaining payoff starting at \((t, y)\) under
\(\lambda\) is
\[
  J(t, y;\lambda) = \mathbb{E}_{t, y}\left[\int_{t}^T e^{-r(u-t)}
    f(u,Y_u,\lambda_u)\,du \right],
\]
and the value function is
\[
  v(t,y) = \sup_{\lambda\in\Lambda_p(t, y)} J(t ,y;\lambda).
\]
We restrict the admissible class further, if necessary, so that the
controlled state process has the finite moments required for the
dynamic-programming and verification arguments below.

\subsection{Dynamic Programming Principle}
\label{sec:dpp}

\begin{lemma}[DPP]
\label{lem:dpp}
  Assume \(V\) is continuous.  For every stopping time \(\theta\) with values in \([t,T]\),
  \[
    v(t, y)
    =
    \sup_{\lambda\in\Lambda_p(t, y)}
    \mathbb{E}_{t, y}\left[
    \int_{t}^{\theta}
    e^{-r(u-t)}f(u,Y_u,\lambda_u)\,du
    +
    e^{-r(\theta-t)}V(\theta,Y_\theta)
    \right].
  \]
\end{lemma}
\begin{proof}
  Fix an admissible control on \([t,T]\) and decompose its payoff at
  \(\theta\). Conditional on \(\mathcal{F}_\theta\), the continuation
  value from \((\theta,Y_\theta)\) is bounded above by
  \(v(\theta,Y_\theta)\), which gives one inequality after taking the
  supremum over admissible controls.

  For the reverse inequality, choose any admissible control up to
  \(\theta\) and concatenate it with an \(\varepsilon\)-optimal
  admissible continuation control for the post-\(\theta\) state
  \((\theta,Y_\theta)\). Stability of \(\Lambda_p\) under concatenation
  and the Markov property give an admissible control on \([t,T]\) whose
  payoff is within \(\varepsilon\) of the displayed right-hand side.
  Sending \(\varepsilon\downarrow0\) proves the result.
\end{proof}

\subsection{Martingale Principle}
\label{sec:martingale_principle}

\begin{lemma}[Martingale Principle]
  \label{lem:martingale_principle}
  For an admissible control \(\lambda\), define the Bellman process
  \[
    M_t^\lambda = \int_{0}^t e^{-ru} f(u,Y_u,\lambda_u)\,du +
    e^{-rt}V(t,Y_t).
  \]
  For \(0\leq r\leq\theta\leq T\),
  \[
    M_r^\lambda \geq
    \mathbb{E}\left[M_\theta^\lambda\mid\mathcal{F}_r\right],
  \]
  so \(M^\lambda\) is a supermartingale. Under an optimal control it
  is a martingale.
\end{lemma}
\begin{proof}
  Apply the DPP at time \(r\) with stopping time \(\theta\). For the
  fixed continuation control \(\lambda\), the value \(V(r,Y_r)\) is at
  least the conditional expected payoff earned from \(r\) to
  \(\theta\) plus the discounted continuation value at \(\theta\).
  Multiplying by the discount factor from \(0\) to \(r\) and adding
  the payoff accumulated on \([0,r]\) gives the displayed
  supermartingale inequality. If \(\lambda\) is optimal, the DPP is
  attained along \(\lambda\), so the inequality is an equality.
\end{proof}

\subsection{HJB equation}
\label{sec:hjb_equation}

To derive the HJB equation from the DPP, assume
\(V\in C^{1, 1}([0, T], \mathbb{R}_+^2)\) and \(f\) is continuous in
\((t, y)\) for each fixed \(\lambda\in L\). For a test function
\(\phi(t,y)\), define the forward differences
\[
  \Delta_s\phi(t, y) = \phi(t,s+\delta,x)-\phi(t,s,x), \qquad
  \Delta_x\phi(t, y) = \phi(t,s,x+\kappa)-\phi(t,s,x)
\]
with \(y = (s, x)\). The controlled state generator acts only on the
state variables:
\[
  \mathcal A^\lambda\phi(t, y) = \lambda\Delta_s\phi(t, y) +
  \mu\Delta_x\phi(t, y).
\]
Note that the time derivative is not part of \(\mathcal A^\lambda\).

Applying the dynamic programming principle over a short interval
\([t,t+h]\), assuming smoothness of \(v\), gives
\[
  0= \sup_{\lambda\in L} \left\{sP(t,s,x)-C(\lambda) + \partial_t
    v(t,s,x) + \mathcal A^\lambda v(t,s,x) - rv(t,s,x) \right\}.
\]
Equivalently,
\[
  \partial_t v(t,s,x) + \sup_{\lambda\in L} \left\{\lambda\Delta_s
    v(t,s,x)-C(\lambda) \right\} + \mu\Delta_x v(t,s,x) + sP(t,s,x) -
  rv(t,s,x) = 0,
\]
with terminal condition \(v(T,s,x)=0\). This is the single-technology
HJB in equation~\eqref{eq:hjb_single}.

\subsection{Verification Theorem}
\label{sec:verification}

We record the verification statement corresponding to the HJB above.

\begin{theorem}[Verification]
  \label{thm:verification}
  Let \(w\in C^{1, 1}([0,T)\times \mathbb{R}_+^2)\cap
  C([0,T]\times\mathbb{R}_+^2)\) have at most polynomial growth, with
  sufficient integrability under the admissible controls to justify
  Dynkin's formula and passage to the terminal time. Suppose
  \[
    \partial_t w(t,y) + \sup_{\lambda\in L} \left\{\mathcal A^\lambda
      w(t,y)-C(\lambda) \right\} + sP(t,y) - rw(t,y) \leq 0
  \]
  on \([0,T)\times\mathbb{R}_+^2\), and \(w(T,y)\geq0\). Then
  \[
    w(t,y)\geq v(t,y).
  \]
  If, in addition, \(w(T,y)=0\) and there exists a measurable selector
  \(\hat\lambda(t,y)\in\mathbb{R}_+\) attaining the supremum such that
  \[
    \partial_t w(t,y) + \mathcal A^{\hat\lambda}w(t,y) -
    C(\hat\lambda(t,y)) + sP(t,y) - rw(t,y) = 0,
  \]
  and the feedback control \(\hat\lambda_u=\hat\lambda(u,Y_{u-})\) is
  admissible, then \(w=V\) and \(\hat\lambda\) is optimal.
\end{theorem}
\begin{proof}
  Fix \((t,y)\) and an admissible control \(\lambda\). Let
  \(Y_u = (S_u,X_u)\) denote the corresponding state process on \([t,T]\).
  Localizing if necessary and applying Dynkin's formula to the
  discounted process gives
  \begin{align*}
    \mathbb E\!\left[e^{-r(\tau-t)}w(\tau,Y_\tau)\right]
    &= w(t,y) \\
    &\quad+ \mathbb E\!\left[\int_t^\tau e^{-r(u-t)} \left(\partial_t w(u,Y_u) + \mathcal A^{\lambda_u}w(u,Y_u) - rw(u,Y_u) \right)\,du \right].
  \end{align*}
  Since the supersolution inequality implies
  \[
    \partial_t w + \mathcal A^\lambda w - rw + sP - C(\lambda) \leq 0
  \]
  for every admissible \(\lambda\), we obtain
  \[
    w(t,y) \geq \mathbb E\!\left[\int_t^\tau e^{-r(u-t)}
      \bigl(S_uP(u,Y_u)-C(\lambda_u)\bigr)\,du +
      e^{-r(\tau-t)}w(\tau,Y_\tau) \right].
  \]
  Letting \(\tau\uparrow T\) and using \(w(T,\cdot)\geq0\) gives
  \(w(t,y)\geq J(t,y;\lambda)\). Taking the supremum over \(\lambda\)
  yields \(w\geq v\).

  If \(\hat\lambda\) attains the supremum and the equality condition
  holds, the preceding inequalities become equalities under
  \(\hat\lambda\), with terminal value zero. Hence
  \(w(t,y)=J(t,y;\hat\lambda)\leq v(t,y)\). Together with \(w\geq v\),
  this proves \(w=v\) and optimality of \(\hat\lambda\).
\end{proof}

\subsection{Multi-technology Extension}
\label{sec:multi-tech_extension}

Analogously to the single-technology setting, we define a probability
space with now \(d\) independent Poisson random measures
\(\mathcal{N}^j(dz, dt)\) \((j\in \{1, \ldots, d\})\) on
\(\mathbb{R}_{+}\times \mathbb{R}_+\) and apply thinning to construct
\(d\) controlled-intensity point processes driving available supply of
each technology. The aggregate available supply \(S_t\), which we
define by superposition, and reference data-center demand \(X_t\)
still constitute the state. For the higher-dimensional case, now
\(L\subseteq \mathbb{R}_{\geq 0}^d\),
\(C\colon \mathbb{R}_{\geq 0}^d\to \mathbb{R}_{\geq 0}^d\), and
\[
f(t, y, \lambda) = s P(t, y) - \sum_{j=1}^d C_j(\lambda_j), \qquad f\colon [0, T]\times \mathbb{R}_+^2\times \mathbb{R}_+^d\to \mathbb{R},
\]
the same formulas for the expected remaining payoff and value function
apply here. The DPP and martingale principle continue to holdt by
straightforward extensions of the arguments presented in the
single-technology case.

For \(d\) supply technologies with controlled intensities
\(\boldsymbol\lambda=(\lambda_1,\ldots,\lambda_d)\) and jump sizes
\(\delta_1,\ldots,\delta_d\), define
\(\Delta_{s_j}\phi(t,y) = \phi(t,s+\delta_j,x)-\phi(t,s,x)\). The
state generator is then
\[
  \mathcal A^{\boldsymbol\lambda}\phi(t,y) = \sum_{j=1}^d
  \lambda_j\Delta_{s_j}\phi(t,y) + \mu\Delta_x\phi(t,y).
\]
The leads the multi-technology HJB presented in
Section~\ref{sec:multiple_technologies},
equation~\eqref{eq:multi_control_hjb}. The verification theorem is a
natural multivariate extension of Theorem~\ref{thm:verification}, so
we do not repeat it here.
 \section{Numerical Implementation Details}
\label{app:numerics}

We first describe the numerical solver for \eqref{eq:hjb_single}.
Provided positive integers \(N_s\) and \(N_x\), the algorithm operates
on a discrete state space
\[
  \{s_{\text{min}} + \delta i : i\in \{0, 1, \ldots, N_s-1\}\}\times
  \{x_{\text{min}} + \kappa j : j\in \{0, 1, \ldots, N_x-1\}\}.
\]
It will be convenient to think of this grid as a (tensor) product of
vectors, i.e.
\[
  \mathbf{s} = \begin{bmatrix} s_{\min} & \cdots & s_{\min} + \delta
    (N_s - 1) \end{bmatrix}^{\top} \quad\text{and}\quad \mathbf{x}
  = \begin{bmatrix} x_{\min} & \cdots & x_{\min} + \kappa (N_x -
    1) \end{bmatrix}^{\top}.
\]
We then discretize the first-order system of ODEs that corresponds to
the discretized version of \eqref{eq:hjb_single}, namely
\begin{equation}\label{eq:discretized_hjb_single_app}
  \frac{dv_{i,j}}{dt} + \lambda_{i,j}^{*}(v_{i+1,j} - v_{i,j}) + \mu(v_{i,j+1} - v_{i,j}) + s_i P(t, s_i, x_j) - C(\lambda_{i,j}^{*}) - rv_{i,j} = 0.
\end{equation}
where \(i\) and \(j\) run through \(\{0, \ldots, N_s-1\}\) and
\(\{0, \ldots, N_x-1\}\) respectively. We implement a semi-implicit
Euler method to solve \eqref{eq:discretized_hjb_single_app}. At each
iteration \(n\), we first compute \(\lambda^{*, (n)}_{i,j}\) by
\eqref{eq:optimal_supply_intensity} with
\(v^{(n)}_{i+1,j} - v^{(n)}_{i,j}\) replacing \(\Delta_s v\). With
\(\lambda^{*, (n)}_{i,j}\) fixed, after substituting a finite
difference for the time derivative, we solve
\begin{multline}
  \label{eq:euler_hjb_single_app}
  v^{(n+1)}_{i,j} - \Delta t\left[\lambda^{*,(n)}_{i,j}(v^{(n+1)}_{i+1,j}-v^{(n+1)}_{i,j}) +\mu(v^{(n+1)}_{i,j+1}-v^{(n+1)}_{i,j}) -rv^{(n+1)}_{i,j} \right]  \\
  = v^{(n)}_{i,j} + \Delta t\left[s_i P^{(n)}(s_i, x_j) - C(\lambda^{*,(n)}_{i,j}) \right]
\end{multline}
for \(V^{(n+1)}\). Because \(V^{0} = V(T, \cdot) = 0\), the solver
works backwards in time, and \(V^{(n+1)}\) is (approximately) the
value at time \(\tau_n - \Delta t\) where \(\tau_n\) is the time
corresponding to \(V^{(n)}\). Observe that this is a linear system.
Indeed, define \(N = N_s\cdot N_x\)-dimensional vectors \(V^{(n)}\),
\(S\) and \(X\) by stacking the entries of \(v^{(n)}\) row by row to
build \(V^{(n)}\) and filling each columns of \(S\) with
\(\mathbf{s}\) and each row of \(X\) with \(\mathbf{x}\). Then
\eqref{eq:euler_hjb_single_app} becomes
\[
  (I - \Delta t M^{\lambda^{*,(n)}}) V^{(n+1)} = V^{(n)} + \Delta
  t\,[S P(t, S, X) - C(\boldsymbol{\lambda}^{*,(n)})].
\]
Here \(I\) is the identity and \(M^{\lambda^{*,(n)}}\) is defined to have entries
\[
  M^{\lambda^{*,(n)}}_{q,q}=-(r+\lambda^{*,(n)}_{i,j}+\mu),\qquad
  M^{\lambda^{*,(n)}}_{q,q+1}=\mu,\qquad
  M^{\lambda^{*,(n)}}_{q,q+N_x}=\lambda^{*,(n)}_{i,j}.
\]
where \(q=iN_x+j\). The entry \(q+1\) is included only when
\(j<N_x-1\), and the entry \(q+N_x\) is included only when
\(i<N_s-1\). At the upper demand and supply boundaries we use zero
forward differences, omit the corresponding off-diagonal transition
and its matching diagonal rate, and do not allow wrap-around between
rows. Notice that \(I-\Delta t M^{\lambda^{*,(n)}}\) is sparse, so we
can solve for \(V^{(n+1)}\) efficiently.

The multi-technology solver generalizes the single-control solver to the
multi-technology setting.  It begins by similarly considering the
discrete state space
\[
  \{ s_{\min} + \eta i : i\in \{0, 1, \ldots, N_s-1\} \} \times
  \{x_{\min} + \kappa j : j\in \{0, 1, \ldots, N_x-1\}\}
\]
Here \(0 < \eta\leq \min_{1\leq \ell\leq d}\delta_\ell\) is a
coarseness parameter chosen small enough so that the \(d\) possible
jump sizes are multiples of \(\eta\). For technology \(\ell\), write
\(w_\ell=\delta_\ell/\eta\). Following the same arguments as in the
single supply case, we derive the sparse linear system
\[
  (I - \Delta t M^{\boldsymbol\lambda^{*,(n)}}) V^{(n+1)} = V^{(n)} +
  \Delta t\left[S P(t, S, X) - \sum_{\ell=1}^d
    C_\ell(\lambda_\ell^{*,(n)})\right].
\]
where \(M^{\lambda^{*, (n)}}\) now has entries
\[
  M^{\boldsymbol\lambda^{*,(n)}}_{q,q} =
  -\left(r+\mu+\sum_{\ell=1}^d\lambda_{\ell,i,j}^{*,(n)}\right),
  \qquad M^{\boldsymbol\lambda^{*,(n)}}_{q,q+1}=\mu, \qquad
  M^{\boldsymbol\lambda^{*,(n)}}_{q,q+w_\ell N_x} =
  \lambda_{\ell,i,j}^{*,(n)},
\]
for \(q=iN_x+j\). In the implementation, these entries are assembled
through precomputed sparse jump matrices, which also handle the same
indexing for interpolation on the supply grid. The interpolation is
done for purely computational purposes as it avoids potential
floating-point arthimetic errors when computing \(w_{\ell}\).
 \printbibliography

\end{document}